\documentclass[journal]{IEEEtran}
\ifCLASSINFOpdf
   \usepackage[pdftex]{graphicx}
\else
\fi
\usepackage[cmex10]{amsmath}
\usepackage[tight,footnotesize]{subfigure}
\usepackage{amsfonts}
\usepackage{graphics} 
\usepackage{setspace}
\usepackage{amssymb}
\usepackage{epsfig} 
\usepackage{stfloats}
\usepackage{multirow}
\usepackage{makecell}
\usepackage{algorithm}
\usepackage{algorithmic}
\usepackage{amsthm}
\floatname{algorithm}{Algorithm}
\usepackage{enumerate}
\usepackage{subeqnarray}
\theoremstyle{plain}
\newtheorem{theorem}{Theorem}
\newtheorem{lemma}{Lemma}
\newtheorem{definition}{Definition}
\usepackage{bm}
\begin{document}
%
\title{SDN-based Resource Allocation in Edge and Cloud Computing Systems: An Evolutionary Stackelberg Differential Game Approach}
\author{Jun~Du,~\IEEEmembership{Member,~IEEE}, Chunxiao~Jiang,~\IEEEmembership{Senior Member,~IEEE}, Abderrahim~Benslimane,~\IEEEmembership{Senior~Member,~IEEE}, Song~Guo,~\IEEEmembership{Senior Member,~IEEE} and Yong~Ren,~\IEEEmembership{Senior Member,~IEEE}
\thanks{J. Du and Y. Ren are with the Department of Electronic Engineering, Tsinghua University, Beijing 100084, P. R. China (e-mail: blgdujun@gmail.com, reny@tsinghua.edu.cn).}
\thanks{C. Jiang is with Tsinghua Space Center, Tsinghua University, Beijing 100084, P. R. China (e-mail: jchx@tsinghua.edu.cn).}
\thanks{A. Benslimane is with the Department of Computer Science, University of Avignon, 84911 Avignon, France (e-mail: abderrahim.benslimane@univ-avignon.fr).}
\thanks{S. Guo is with the Department of Computing, The Hong Kong Polytechnic University, Hong Kong (e-mail: song.guo@polyu.edu.hk).}
\vspace{-5mm}
}
\maketitle

\begin{abstract}
Recently, the boosting growth of computation-heavy applications raises great challenges for the Fifth Generation (5G) and future wireless networks.
As responding, the hybrid edge and cloud computing (ECC) system has been expected as a promising solution to handle the increasing computational applications with low-latency and on-demand services of computation offloading, which requires new computing resource sharing and access control technology paradigms.
This work establishes a software-defined networking (SDN) based architecture for edge/cloud computing services in 5G heterogeneous networks (HetNets), which can support efficient and on-demand computing resource management to optimize resource utilization and satisfy the time-varying computational tasks uploaded by user devices.
In addition, resulting from the information incompleteness, we design an evolutionary game based service selection for users, which can model the replicator dynamics of service subscription.
Based on this dynamic access model, a Stackelberg differential game based cloud computing resource sharing mechanism is proposed to facilitate the resource trading between the cloud computing service provider (CCP) and different edge computing service providers (ECPs).
Then we derive the optimal pricing and allocation strategies of cloud computing resource based on the replicator dynamics of users' service selection. These strategies can promise the maximum integral utilities to all computing service providers (CPs), meanwhile the user distribution can reach the evolutionary stable state at this Stackelberg equilibrium. Furthermore, simulation results validate the performance of the designed resource sharing mechanism, and reveal the convergence and equilibrium states of user selection, and computing resource pricing and allocation.
\end{abstract}


%
\IEEEpeerreviewmaketitle

\section{Introduction}

Recently, computation-heavy applications are experiencing a dramatic increasing over the Fifth Generation (5G) and future wireless networks.
There is evidence that such applications, including mining process for Proof-of-Work (PoW) in blockchain, interactive gaming, virtual reality, video services, etc., have become premier drivers of the exponential computing task growth~\cite{jiang2016machine, nguyen2019market, yang2019communication}.
To handle such increasing computing requirements, hybrid edge and cloud computing (ECC) systems have been expected to provide low-latency and on-demand computing services to users~\cite{chen2015efficient, jiao2019auction, du2020learningVTM}.
In ECC systems, cloud computing, as the traditional solution of computation offloading for user devices, is usually implemented at cloud nodes physically located far from users, which results in a long latency service response.
Aiming at this problem, edge computing has been proposed as the complement of cloud computing by enabling users to upload computational tasks to the edge of networks~\cite{premsankar2018edge, jovsilo2018decentralized}, which can eliminate the latency and enhance the reliability of services.
However, with the growing amount of computational task requirements, computational power limited edge servers might be overwhelmed with severe performance degradation. A feasible solution for this problem is forwarding these tasks at edge nodes to the remote cloud center~\cite{liu2018distributed, du2021oceanSHS}, which can be considered as computation offloading between edge computing service providers (ECPs) and the cloud computing service provider (CCP).
Therefore, to achieve the optimal and stable performance of CCP systems, an efficient cloud computing resource sharing mechanism plays an important role resulting from the constrained resource equipped by the CCP and time-varying user requirements among the CCP system.
In addition, such mechanism is more challenging when the dynamic service subscription of users is taken into account~\cite{laaroussi2018service}.
This work will establish a hybrid ECC system, in which users can upload their computational tasks to nearby ECPs or the remote CCP dynamically.
In addition, by considering the dynamic service subscription of users among the CCP and ECPs, this work will focus on the computing resource sharing and computation offloading mechanism design in the ECC system to realize an efficient utilization of computing resource and satisfy the service requirements of users.

As mentioned previously, the mobility and time varying service selection of users may bring difficulties to efficient resource sharing mechanism designs.
In addition, there always exist bidirectional data interactions, including service subscription, task uploading, service response, etc., between end users and computing servers located at either the CCP or ECPs in the ECC system.
These frequent interactions may lead to congestions at different computing providers (CPs)~\cite{chaudhary2017network}.
To solve these problems, an appropriate network architecture is necessary to realize an effective management of the hybrid ECC system.
In recent years, software-defined networking (SDN) has been considered as an advanced network architecture to achieve flexible resource management and system performance control~\cite{cao2017enhancing, du2018auction},
which can mitigate challenges above.
Moreover, taking advantage of the available and accurate information of global system status collected by the SDN controller, the system
can make optimal decisions to improve resource utilization and service quality~\cite{liang2018enhancing}.
On the other hand, latency problems, fault and Disruption tolerance, and scalability issues brought by the SDN-based fully centralized control architecture can be well solved by the integrated cloud and edge computing mechanism.
Therefore, in this work, an SDN-based architecture will be established for computing resource sharing and computation offloading in the ECC system.
With a centralized controller, SDN will help CPs to dynamically adjust the resource sharing and computation offloading strategies, which can match time-varying demands of users by observing their dynamic service selection.

Considering that the SDN-based fully centralized control architecture established in Section~\ref{sdwn} will suffer from latency problems, fault and Disruption tolerance, and scalability issues, this section will introduce an ECC system to realize
\vspace{-2mm}

\subsection{Related Work}


For the integrated ECC system, the computation offloading mechanism plays a crucial role in improving resource utilization and service quality.
Such computation offloading involves two aspects. Specifically, in the aspect of users, both the CCP and ECPs can offload users' computational tasks with different processing latency and transfer latency. On the other hand, resulting from the limited computational power equipped, ECPs are not qualified for providing services of heavy-computation tasks processing. Then ECPs have to forward some of these tasks to the remote CCP, which has powerful and dedicated computing resource and can provide services on demand. Such process above can be also considered as computation offloading between the CCP and ECPs.
According to such two-layered resource sharing among different CPs and users, how to allocate cloud computing resource among ECPs and users selecting the CCP will influence the resource utilization and service quality significantly.

Driven by the supply and demand of computing resource among the CCP, ECPs and users, the resource trading can be formed and facilitated, which needs to satisfy the demands of users selecting different CPs, and meanwhile maximize the utility of each CP.
For these purposes, many researches have focused on effect and efficient resource allocation and sharing mechanisms in edge/cloud systems, by introducing different economic models based on auction~\cite{kiani2017toward, zhang2019parking}, contract~\cite{wang2018knowledge, du2017contractJSAC}, Stackelberg game~\cite{ Aujla2018Optimal, Zhang2017Computing, Xiong2017Edge}, etc..
Among these studies, auction and contract based trading mechanisms were designed to motivate participants to report their service requirements or capacities truthfully, which can deal with trustworthiness and information asymmetric issues in the system. On the other hand, Stackelberg game provides a suitable framework to model the interactions of trading strategies made in supply and demand sides, including resource pricing, requests and proving for communications~\cite{8004162}, storage~\cite{8880515}, energy~\cite{7949095}, etc., which can facilitate the resource trading efficiently and dynamically.
In~\cite{Aujla2018Optimal}, a multi-leader multi-follower Stackelberg game was studied to provide cost-effective migrations of data centera in edge-cloud environment.
To optimize resource allocation of all cloud and fog computing nodes, a Stackelberg game was formulated in~\cite{Zhang2017Computing}, in which fog computing relied on a set of low-power fog nodes that were located close to the end users to offload the services originally targeting at cloud centers.
A two-stage Stackelberg game was introduced into the blockchain consensus process in order to incentive the cooperation between the edge/cloud providers and the miners in a PoW-based blockchain network~\cite{Xiong2017Edge}.
Similarly, in some current studies, Stackelberg game frameworks were also formulated to model the interaction between the edge/cloud nodes and users~\cite{Feng2017Dynamic, Kim2018An}.
However, all these studies above only considered the computing resource trading between users and CPs, or between the CCP and ECPs, while interactions and influences among the three levels were hardly investigated.
In fact, users' service subscription will impact computation offloading between the CCP and ECPs. In addition, computing service qualities received by users selecting different CPs will vary with different resource sharing strategies made by CPs. Then users will change their selection strategies for better services, considering that users are rational.
It is difficult to model and analyze these interactions above, since that the strategies made by the three parties will impact and be impacted by each other.
To solve such interactive issues, this work will establish an evolutionary game based model to analyze the users' dynamic service selection among CPs. In addition, we will propose a Stackelberg differential game based cloud computing resource sharing mechanism, which will dynamically determine the optimal resource pricing and allocation/request strategies for the CCP and ECPs.
In this Stackelberg differential game, the differential equation is introduced based on the replicator dynamics of user selections, which can establish the connection between evolutionary game operated among users and Stackelberg differential game operated among CPs. According to such hierarchical control and optimization, computing resource utilization and user service quality can be both improved.

As mentioned previously, dynamic user selections will bring challenges to the optimization of resource sharing. Such joint optimization for users and heterogeneous CPs can be implemented
efficiently by an SDN-based architecture.
The SDN-based architecture design for the ECC systems has attracted researchers' great attention, especially in the 5G heterogeneous environment and various Internet of Things (IoT) applications~\cite{molina2018enhancing, farris2018survey, AUJLA20181279}.
To realize efficient and secure resource management, data processing and access control, different SDN-based architectures have been investigated.
In~\cite{Bruschi2019A}, authors introduced a tunnel-less SDN scheme for scalable realization of virtual tenant networks across the 5G heterogeneous infrastructure, which could support migrations of software instances among geo-distributed computing resources.
To meet requirements of various applications and improve the end-to-end system performance efficiently, a novel integrated framework including SDN, computing, and caching was designed in~\cite{chen2018joint}.
In~\cite{Baktir2017How}, the cooperation among edge computing nodes was investigated, and their interactions were realized by establishing an SDN related mechanism.
An SDN-based control scheme was designed in~\cite{8008830} for a multi-edge-cloud environment involves huge data migrations to realize an efficient traffic flow scheduling.
In addition, different SDN-based distributed and layered network architectures were also investigated to operate edge/cloud computing systems with blockchain techniques, which can deal with problems brought by limited bandwidth, high latency, large volume of data, and real-time analysis requirements~\cite{Sharma2017DistBlockNet, Sharma2018Energy}.
In summary, leveraging the SDN-based architecture, flexible management of heterogonous resource and optimal control of system performance can be implemented to support computation-heavy applications. Therefore, this work will design an SDN-based architecture to optimize the computing resource utilization by considering the dynamic users' service subscription.
\vspace{-2mm}

\subsection{Contributions and Organization}

Main contributions of this paper are summarized as follows.
\begin{itemize}
  \item We establish an SDN-based architecture for computing resource sharing in the ECC system. Taking advantage of SDN controllers which separate the distributed infrastructure and resource management, the dynamic optimal pricing and allocation strategies can be obtained.
  \item We design a Stackelberg differential game based cloud computing resource sharing, which determines the optimal resource pricing and allocation/request strategies dynamically. Then an open-loop Stackelberg equilibrium is derived as the optimal solution. Comparing with traditional static strategies, the proposed mechanism can achieve higher integral utilities in a time horizon and faster convergence speed of decision making.
  \item We propose a hierarchical dynamic game framework composed of evolutionary game in the user layer and Stackelberg differential game in the edge and cloud layer, which can incentive the cooperation of cloud computing resource sharing. Different from the traditional separated control outside the user layer, this work considers the dynamic service selections of users among edge and cloud resource. Based on this consideration, the user service requirements can be satisfied as well as the edge/cloud computing resource can be utilized efficiently.
  \item We analyze the performance of the designed computing resource sharing mechanism based on the hierarchical dynamic game. Specifically, the existence and uniqueness of equilibrium of user selections, as well as their evolutionary stable states, are analyzed. In addition, the optimal dynamic pricing and allocation of cloud computing resource are derived based on the replicator dynamics of users' service selection. Furthermore, simulation results validate the performance of the designed resource sharing mechanism, and reveal the convergence and stable states of user selection, resource pricing and resource allocation.
\end{itemize}

The rest of this paper is organized as follows. The SDN-based architecture for the ECC system is established in Section~\ref{sdwn}. Section~\ref{system} presents the system model and proposed hierarchical game framework.
An evolutionary game for service selection of user devices is designed in Section~\ref{evolutionary}, and the Stackelberg differential game based computing resource pricing and allocation schemes are proposed in Section~\ref{stackelberg}.
Simulations are shown in Section~\ref{simulation},~and conclusions are drawn in Section~\ref{conculsion}.
\vspace{-2mm}

\section{SDN Architecture Design for Edge/Cloud Computing Systems}
\label{sdwn}

This section will propose a model of layered edge/cloud computing service providing system which takes advantage of SDN paradigms, and show that how to implement such cross-layer computing service for users in the designed SDN-based architecture in detail.
The SDN-based architecture, which consists of three levels, i.e., the infrastructure plane, control plane and management plane, is established based on the infrastructure in 5G wireless heterogeneous networks (HetNets), as shown in Fig.~\ref{fig:SDN}. According to such SDN-based management, CPs will provide edge and cloud computing services to users.
In this architecture, the cloud computing center and edge computing nodes are operated by the CCP and different ECPs, respectively. These CPs constitute the infrastructure plane and take charge of providing computing services to different types of user devices.
In addition, we consider a user layer outside the SDN architecture, in which user devices such as mobile phones, tablets, etc., can receive computing services by subscribing to different CPs, according to their computational tasks' requirements.
Such architecture can realize a dynamic and real-time information collection of task requirements and system workload, and decision making of computation offloading.
Next, we will design the function and operation of the three planes in order to satisfy the computing requirements of users, and meanwhile manage the computing resource sharing and providing efficiently.
\vspace{-3mm}

\begin{figure}[!t]
  \centering
  \includegraphics[width=0.4\textwidth]{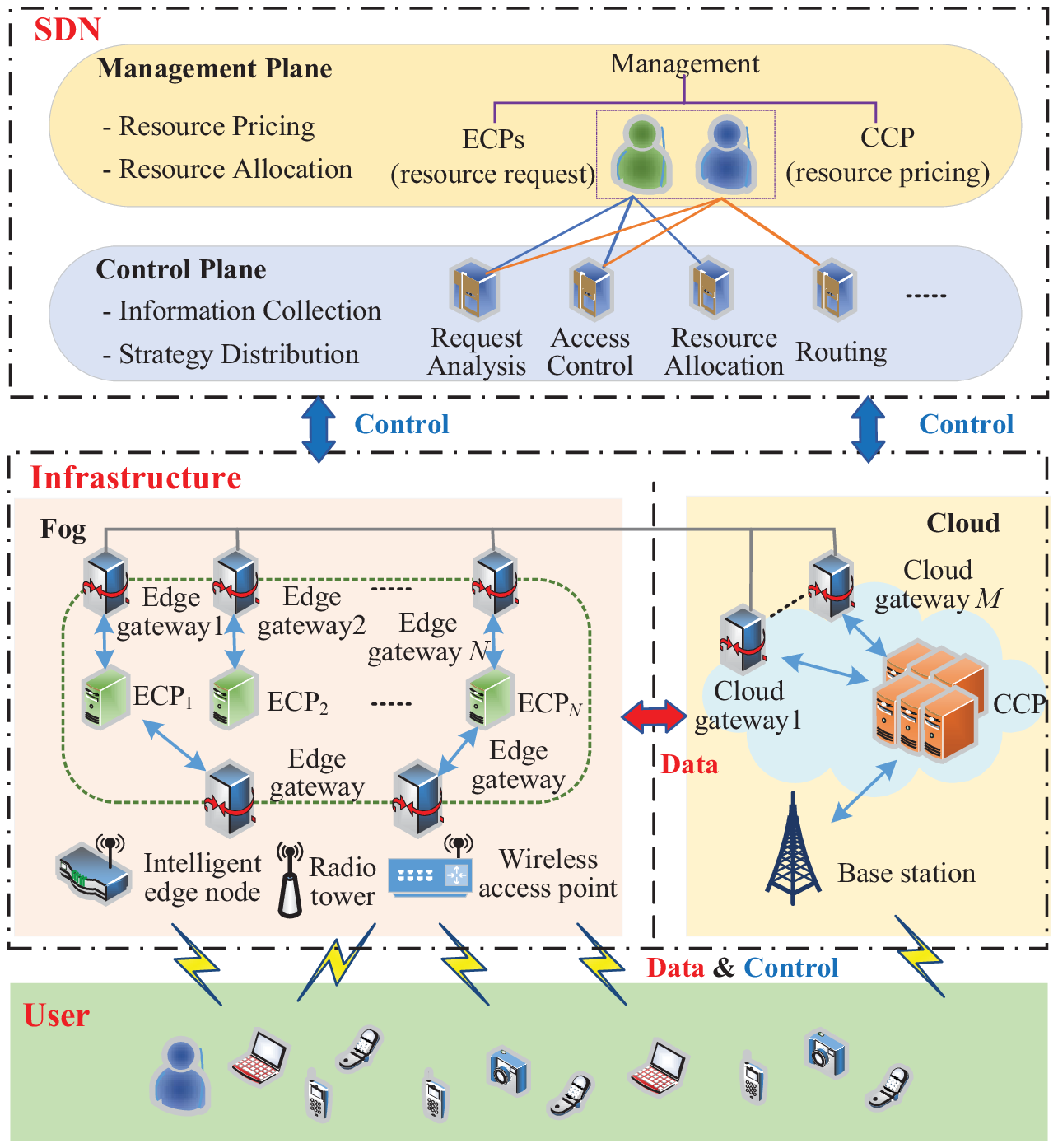}\vspace{-2mm}
  \caption{Architecture of SDN-based resource pricing, sharing and user scheduling in an ECC system.}\vspace{-5mm}
  \label{fig:SDN}
\end{figure}

\subsection{Infrastructure Plane}
In the infrastructure plane, the CCP and ECPs provide remote and edge computing services for users, respectively, who can access these CCP and ECPs via edge nodes, such as wireless access points, radio towers, and macro-cell base station in 5G HetNets.
In addition, surveillance cameras and servers at the edge of Radio Access Networks (RANs) can be also operated as edge servers to provide ubiquitous computing services.
In current designed mobile computing systems, edge nodes can connect to different ECPs through edge gateways and Device-to-Device (D2D) communication without effecting the backhaul network~\cite{chaudhary2017network}. By subscribing to these ECPs, users can receive fast response with respect to computing services.
Alternatively, users can also select the remote CCP, who usually possesses richer computational power to provide higher-speed computing services than ECPs.

Considering the limited computational power of ECPs, the CCP can share parts of its computational power with ECPs through wireline connections between the edge gateways and cloud gateways.
Resulting from the dynamic service selection of users, how to allocate the cloud computing resource among ECPs and users selecting the CCP will influence the efficiency of ECC system.
Therefore, an efficient and dynamic control on computing sharing between the CCP and ECPs plays an important role on optimizing the resource utilization and satisfying the resource requirements of users with fast response.
\vspace{-5mm}

\subsection{Control Plane}

As shown in Fig.~\ref{fig:SDN}, the SDN-based architecture separates computing resource management from the infrastructure, which forms a hierarchical game based market of cloud computing resource in the control plane.
In this plane, information collection of users' subscriptions and strategy distribution for CPs in the infrastructure plane will be implemented through the data information interaction between the infrastructure plane and control plane.

\subsubsection{Information Collection}

Through the controller of request analysis, the control plane collects the information of users' subscriptions among different CPs, as well as the local computational power of each CP, and then sends such received information to the management plane in real time. Specifically, the controllers of request analysis and access control are able to communicate with the CCP and ECPs through access points, and then call for the time-varying number of users selecting the corresponding CP.

\subsubsection{Strategy Distribution}

The control plane receives pricing and request strategies of cloud computing resource made by the upper management plane, and then distributes these strategies to the CCP and ECPs through the controller of resource allocation. In addition, the access controller and core controller are responsible for resource sharing of cloud resource and access control between edge gateways and cloud gateways, i.e., setting approach network paths between gateways, managing computation offloading and so on.
\vspace{-3mm}

\subsection{Management Plane}

After receiving the information of users' subscription and computational power equipped at each CP, the optimal pricing and request strategies of cloud resource will be determined at the management plane. To be specific, the management will help the CCP to make decisions on dynamic resource pricing, meanwhile help ECPs to determine how much computational power should be requested.
Then these decisions will be returned back to the control plane, and then guide the cloud computing resource sharing between the CCP and ECPs.
Such service response above can be implemented fast considering the the wire connections between SDN controllers and edge and cloud gateways, and the powerful processing capacities of SDN servers.
\vspace{-3mm}

\section{System Model and Hierarchical Game Framework}
\label{system}

Considering that the SDN-based fully centralized control architecture established in Section~\ref{sdwn} will suffer from latency problems, scalability issues, etc., this section will introduce a hierarchical game based computation offloading mechanism to realize real-time and latency-sensitive services close to users.
\vspace{-5mm}

\subsection{System Model}

Consider an ECC system with a finite set $\mathsf{\mathcal{N}}=\left\{ 1,2,\cdots ,N \right\}$ ECPs overlaying with one remote CCP and providing edge computing service to user devices, including mobile phones, wearable devices, tablets, etc., as shown in Fig.~\ref{fig:system_model}.
In this work, we consider that the CCP and ECPs can be operated by the mobile network operators (MNOs) or provide the computing service as third parties.
These user devices can access and send computational tasks to different ECPs by communicating with edge access points, such as intelligent edge nodes, radio towers, etc., which can upload these computational tasks or service requests to ECPs through edge gateways.
In addition, to fulfil some high-complex computational tasks, users can also request the cloud computing resource through base stations and edge access points, which will upload their computational tasks to the CCP.
Accordingly, the CCP and ECPs will respond to these computing requests on demand.
As compensation, each users needs to pay the CCP or ECPs for accessing fee.
Denote ${{p}_{n}}$ as the price charged by ECP $n$ ($n\in \mathsf{\mathcal{N}}$) and ${{p}_{c}}$ as the price charged by the CCP, which are fixed access fees per device per unit of time paid by users, considering current charging models set by mobile service operators~\footnote{Price ${{p}_{n}}$ and ${{p}_{c}}$ can be set by the MNO or the third-party service providers, depending on who operate these CCP and ECPs.}.

\begin{figure}[!t]
  \centering
  \includegraphics[width=0.46\textwidth]{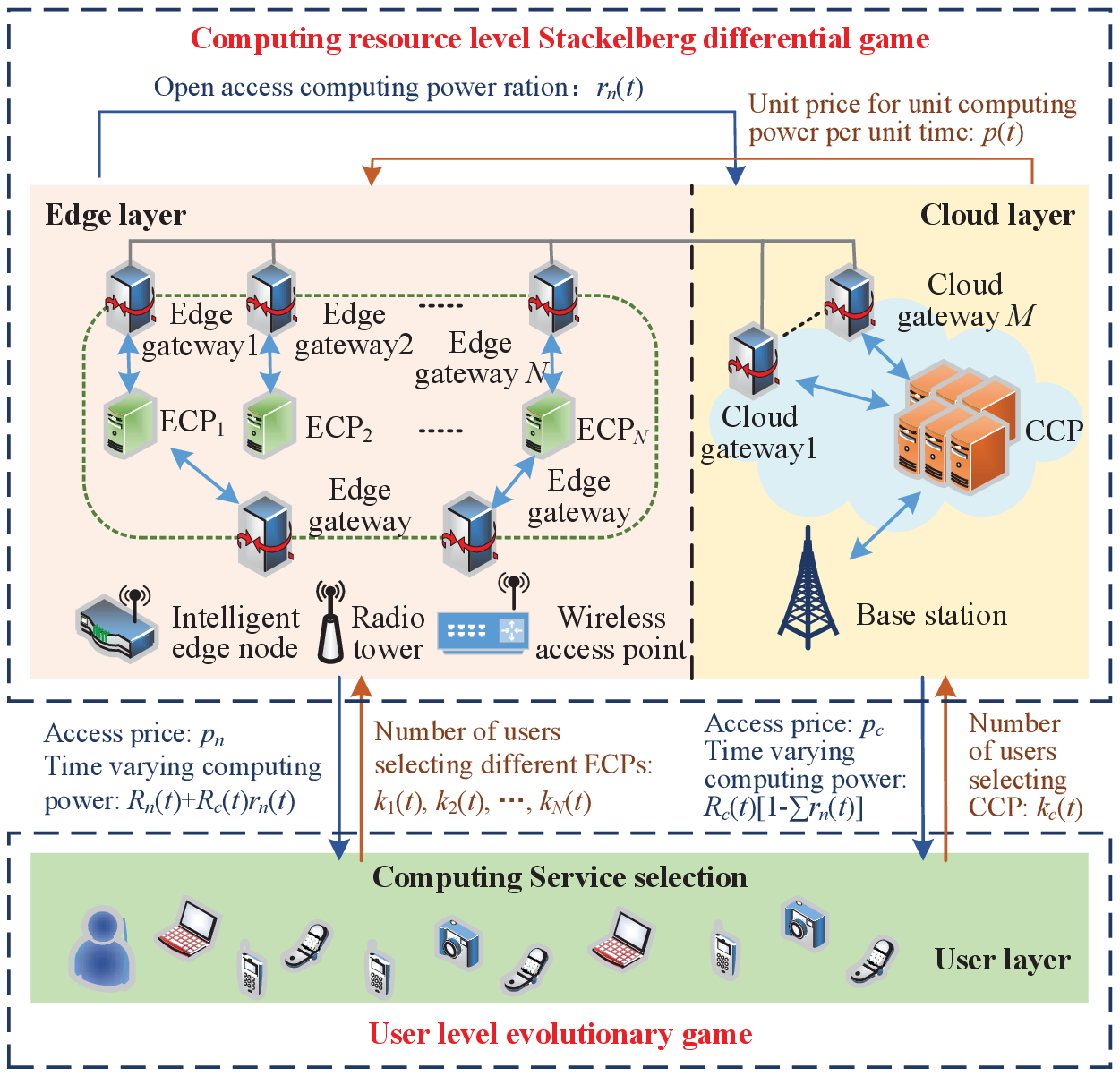}\vspace{-2mm}
  \caption{Hierarchical game based resource pricing and sharing in ECC systems.}\vspace{-3mm}
  \label{fig:system_model}
\end{figure}

In this system, ECPs, performing as light-weight computing servers, might be deployed at base stations, wireless access points, etc., and then providing computing services with shortened latency.
In addition, the CCP, which can promise powerful and stationary computing services, connects to the $N$ ECPs via cloud gateways in a wireline manner.
When ECPs cannot fulfill the received computational tasks under their constrained computing resource, they have to forward some part of computational tasks to the remote CCP through wireline backhaul links. Such computation offloading means that ECPs needs to request the CCP for additional computational power.
For the shared computing resource and possible energy cost resulting from resource allocation and information exchange for the CCP, ECPs need to pay the CCP with uniform price $p$ in monetized payment unit computational power per unit time. In this work, $p$ is a function of time and can be denoted by $p\left( t \right)$. Once announced by the CCP, $p\left( t \right)$ can be observed by all ECPs in the system through the SDN architecture.

After observing price $p\left( t \right)$, ECP $n$ ($n\in \mathsf{\mathcal{N}}$) decides to request ${{r}_{n}}\left( t \right)\in \left[0,1\right)$ proportion of CCPs computational power~\footnote{In this work, the computational power is considered as the computing frequency or speed, which can be measured by the computing times per unit time. In blockchain applications, the unit of computational power can be defined as $H/s$.} during a continuous observation period $\left[ 0,T \right]$. Then the computing resource allocation state of the system can be described by vector $\mathbf{r}={{\left[ {{r}_{1}},{{r}_{2}},\cdots ,{{r}_{N}} \right]}^{T}}$.
Let ${{r}_{c}}\left( t \right)$ denote the current remaining computational power level of the CCP at time $t$. Then we have $\sum\nolimits_{n=1}^{N}{{{r}_{n}}\left( t \right)}+{{r}_{c}}\left( t \right)=1$ for $t\in \left[ 0,T \right]$.
In this resource transaction between the CCP and ECPs, the CCP, as the computing resource provider, needs to optimize unit price $p\left( t \right)$ to
make its resource bring into maximal efficacy; on the other hand, the ECPs, who are the resource receivers and buyers, will make the optimal decision on how many resources to buy to create a tradeoff between its quality of service and cost.

Consider that the local computational power of ECP $n$ is ${{R}_{n}}$ for all $n\in \mathsf{\mathcal{N}}$, the total computational power of CCP is ${{R}_{c}}$ (${{R}_{c}}>{{R}_{n}}$, $\forall n\in \mathsf{\mathcal{N}}$), and the number of all user devices in the system is denoted by $K$.
In addition, let ${{k}_{n}}\left( t \right)$ and ${{k}_{c}}\left( t \right)$ be the numbers of user devices subscribing to ECP $n$ and the CCP at time $t$, respectively. Then ${{x}_{n}}\left( t \right)={{{k}_{n}}\left( t \right)}/{K}$ (${{x}_{n}}\in \left[ 0,1 \right]$) and ${{x}_{c}}\left( t \right)={{k}_{c}}\left( t \right)/K$ represent the population share of ECP $n$ and the CCP accordingly, and we have $\sum\nolimits_{n=1}^{N}{{{x}_{n}}}\left( t \right)+{{x}_{c}}\left( t \right)=1$~\footnote{In this work, we only consider the users having the requirement of uploading their conputational tasks to the CCP or ECPs. Therefore, there exist at least one computing service provider to be selected by users.}.
Based on these definitions above, the computational power allocated to each user selecting ECP $n$ and the CCP can be calculated as
\begin{subequations}
\label{equ:user_utility}
\begin{align}
&{{\omega }_{n}}\left( \mathbf{x},\mathbf{r} \right)=\frac{{{R}_{n}}+{{R}_{c}}{{r}_{n}}\left( t \right)}{K{{x}_{n}}\left( t \right)}, n\in \mathsf{\mathcal{N}},  \\
& {{\omega }_{c}}\left( \mathbf{x},\mathbf{r} \right)=\frac{{{R}_{c}}\left( 1-\sum\nolimits_{n=1}^{N}{{{r}_{n}}\left( t \right)} \right)}{K{{x}_{c}}\left( t \right)},
\end{align}
\end{subequations}
respectively, where vector $\mathbf{x}={{\left[ {{x}_{1}},{{x}_{2}},\cdots ,{{x}_{N}},{{x}_{c}} \right]}^{T}}$ is the population distribution state of the ECC system.

\subsection{Hierarchical Game Framework}
Based on the computing resource market model above, the service selection of user devices, pricing strategy of the CCP and computational power requests of ECPs are time-varying and interact with each other. To be specific, if too many users select the same CP, the received computational power for each of these users will decrease according to (\ref{equ:user_utility}). Then these users tend to leave to other CPs with high average computational power. In addition, CPs (including the CCP and ECPs) with large number of user devices will expect to obtain much computational power to satisfy the computing requirements from accessing user devices. However, such resource request behavior also depends on and will further influence the computing resource price decided by the CCP.
To model interactions among users, ECPs and the CCP analyzed above, this work designs a hierarchical dynamic game based scheme to improve the computing service quality and facilitate the computing resource sharing among CPs in the ECC system.

There are two levels in the hierarchical dynamic game designed for the computing service selection, pricing and sharing system, i.e., the user level and the computing resource level.
In the user level, the behavior of users' dynamic service selection is formulated and analyzed through an evolutionary game model. Then to model the computing service providing and requirement between the CCP and ECPs, a Stackelberg differential game will be designed to optimize the pricing and sharing strategies for the limited computational power.
Such hierarchical dynamic game framework is shown in Fig.~\ref{fig:system_model}.

\subsubsection{Evolutionary Game in User Level}

According to access price $p_c$ / $p_n$ released by the CCP or different ECPs, as well as the received computational power, each user makes its computing service selection among these CPs to improve its utility. For user devices, the access prices are fixed and stay the same over time. On the other hand, one can notice that the received computational power for users is time varying, which depends not only on the number of user devices accessing the same CP currently, but also the dynamic computing resource sharing of CCP among different CPs.
Therefore, without complete information, users can hardly make the optimized service selection globally, i.e., among all CPs and over all time duration.
As a response, each user device will learn and adapt its selection strategy gradually. To model this learning and adaptation process, an evolutionary game can be designed to describe and analyze the dynamic user behavior. Through replicator dynamics, all user devices in the ECC system will reach the same individual utility at the equilibrium~\cite{zhu2014pricing, du2018communityTIFS}.

\subsubsection{Stackelberg Differential Game in Resource Level}

To respond to users' computational task requests on demand, computational power limited ECPs need to buy more computing resource from the CCP. In addition, to compensate the potential loss of accessing users and cost resulting from resource sharing, the CCP will charge corresponding ECP with time-varying unit price $p\left( t \right)$ for the provided computational power. Such pricing strategy is dynamic according to the dynamic computational power sharing/requests and number of accessing user devices at the CCP.
Accordingly, given unit price $p\left( t \right)$, ECPs control their resource requests dynamically to maximize their own utilities.
To analyze such dynamic computing resource pricing of CCP and dynamic resource requests of ECPs, this work establishes a non-cooperate Stackelberg differential game, in the two levels of which, all CPs optimize their own strategies to receive maximized utilities.
In this Stackelberg game, the CCP performs as the leader and ECPs are followers. These players in the computing resource level makes their own optimal decisions dynamically according to the time-varying service selection decisions of user devices.

\section{Evolutionary Game for Service Selection of User Devices}
\label{evolutionary}

In the ECC system, online users with heavy computational tasks will compete for the limited computing resource by selecting and accessing different CPs. Initially, every user selects a candidate CP randomly or by experience. Then to achieve a better service quality, i.e., large computational power and/or low access price, these users
will adapt their selection decisions periodically based on the dynamic received computational power, access price, and the population distribution of all users among different CPs.
During this adaptation process, users cannot optimize their selection strategies globally, resulting from the asymmetry of information. Therefore, to improve their utilities, users will learn by imitating those selection strategies with high utilities gradually.

Evolutionary game can be expected as a suitable tool for modeling such learning and imitating process. Thus this section will first formulate the evolutionary game framework and
evolutionary strategy adaptation for user selection dynamics. Then the evolutionary equilibrium and evolutionary stable state (ESS) will be investigated for the established model.

\subsection{Evolutionary Game based Service Selection}

We first formulate the computing service selection of users among different CPs as an evolutionary game model.

\subsubsection{Players}
The set of $K$ user devices in the service area are the players of the evolutionary game.

\subsubsection{Strategy}
$\mathsf{\mathcal{S}}=\left\{ 1,2,\cdots ,n,\cdots ,N,c \right\}$, where $n\in \mathsf{\mathcal{N}}$ indicates selecting ECP $n$ for computing service, and $s=c$ means selecting the CCP directly.

\subsubsection{Population States}
The population shares of all ECPs constitute the population distribution state denoted by vector $\mathbf{x}={{\left[ {{x}_{1}},{{x}_{2}},\cdots ,{{x}_{N}},{{x}_{c}} \right]}^{T}}\in \mathbb{X}$, where $\mathbb{X}$ represents the state space which contains all possible population distributions.

\subsubsection{Utility}
Given computing resource allocation state $\mathbf{r}={{\left[ {{r}_{1}},{{r}_{2}},\cdots ,{{r}_{N}} \right]}^{T}}$ and population distribution state $\mathbf{x}={{\left[ {{x}_{1}},{{x}_{2}},\cdots ,{{x}_{N}},{{x}_{c}} \right]}^{T}}$, the utility function of user device selecting ECP $n$ and the CCP are defined by
\begin{subequations}
\label{equ:payoff_user}
\begin{align}
&{{\pi }_{n}}\left( n,\mathbf{x},\mathbf{r} \right)=\frac{\beta {{\omega }_{n}}\left( \mathbf{x},\mathbf{r} \right)}{{{p}_{n}}}=\frac{\beta \left( {{R}_{n}}+{{R}_{c}}{{r}_{n}}\left( t \right) \right)}{K{{p}_{n}}{{x}_{n}}\left( t \right)}, \\
& {{\pi }_{c}}\left( c,\mathbf{x},\mathbf{r} \right)=\frac{\beta {{\omega }_{c}}\left( \mathbf{x},\mathbf{r} \right)}{{{p}_{c}}}=\frac{\beta {{R}_{c}}\left( 1-\sum\nolimits_{n=1}^{N}{{{r}_{n}}\left( t \right)} \right)}{K{{p}_{c}}{{x}_{c}}\left( t \right)},
\end{align}
\end{subequations}
respectively,
$\beta >0$ is a constant denoting the mapping factor.

\subsubsection{Replicator Dynamic}
The replicator dynamic reflects the evolutionary behavior of the population among different strategies, i.e., selecting different CPs over time. In this work, we introduce the definition in~\cite{zhu2014pricing, taylor1978evolutionary, romero2018dynamic}, and then give the replicator dynamic as follows,
\begin{subequations}
\label{equ:population_dynamic}
\begin{align}
&\begin{aligned}
{{\dot{x}}_{n}}\!\left( t \right)\!=\!\delta {{x}_{n}}\!\left( t \right)\!\left[ \pi \!\left( n,\mathbf{x}\!\left( t \right),\mathbf{r}\!\left( t \right) \right)\!-\!\pi\! \left( \mathbf{x}\left( t \right),\mathbf{x}\!\left( t \right),\mathbf{r}\!\left( t \right) \right) \right], \\
n\in \mathsf{\mathcal{N}};
\end{aligned}\\
&{{\dot{x}}_{c}}\!\left( t \right)\!=\!\delta {{x}_{c}}\!\left( t \right)\!\left[ \pi\! \left( c,\mathbf{x}\!\left( t \right),\mathbf{r}\!\left( t \right) \right)\!-\!\pi \!\left( \mathbf{x}\!\left( t \right),\mathbf{x}\!\left( t \right),\mathbf{r}\!\left( t \right) \right) \right],
\end{align}
\end{subequations}
with initial population distribution state $\mathbf{x}\left( 0 \right)={{\mathbf{x}}_{0}}\in \mathbb{X}$ ($\forall n\in \mathsf{\mathcal{S}}$), where constant $\delta >0$ is the learning rate of the population which controls the frequency of strategy adaptation for service selection. Moreover, in (\ref{equ:population_dynamic}),
\begin{eqnarray}
\begin{aligned}
 \pi \left( \mathbf{x}\left( t \right),\mathbf{x}\left( t \right),\mathbf{r}\left( t \right) \right) = & \sum\nolimits_{n=1}^{N}{{{x}_{n}}\left( t \right)\pi \left( n,\mathbf{x}\left( t \right),\mathbf{r}\left( t \right) \right)} \\
 & +{{x}_{c}}\left( t \right)\pi \left( c,\mathbf{x}\left( t \right),\mathbf{r}\left( t \right) \right)
\end{aligned}
\end{eqnarray}
is the expected utility of the population given population distribution state $\mathbf{x}\left( t \right)$ and computing resource allocation state $\mathbf{r}\left( t \right)$.
Based on the definitions above, we have
\begin{subequations}
\label{equ:x_dynamic}
\begin{align}
&\begin{aligned}
   {{{\dot{x}}}_{n}}&\left( t \right)=\frac{\delta \beta }{K}\left[ \frac{{{R}_{n}}+{{R}_{c}}{{r}_{n}}\left( t \right)}{{{p}_{n}}} \right. \\
 & \left. \!-{{x}_{n}}\!\left( t \right)\!\left( \sum\limits_{m=1}^{N}\!{\frac{{{R}_{m}}\!+\!{{R}_{c}}{{r}_{m}}\!\left( t \right)}{{{p}_{m}}}}\!+\!\frac{{{R}_{c}}\!\left( 1\!-\!\sum\nolimits_{m=1}^{N}\!{{{r}_{m}}\!\left( t \right)} \right)}{{{p}_{c}}} \!\right)\! \right]\!,
\end{aligned}
\\
& \begin{aligned}
 {{{\dot{x}}}_{c}}&\left( t \right) =\frac{\delta \beta }{K}\left[ \left( \frac{{{R}_{c}}}{{{p}_{c}}}-\sum\limits_{m=1}^{N}{\frac{{{R}_{c}}{{r}_{m}}\left( t \right)}{{{p}_{c}}}} \right) \right.\\
 & \left. \!-{{x}_{c}}\!\left( t \right)\!\left( \sum\limits_{m=1}^{N}{\frac{{{R}_{m}}\!+\!{{R}_{c}}{{r}_{m}}\!\left( t \right)}{{{p}_{m}}}}\!+\!\frac{{{R}_{c}}\!\left( 1\!-\!\sum\nolimits_{m=1}^{N}\!{{{r}_{m}}\!\left( t \right)} \right)}{{{p}_{c}}}\! \right)\! \right]\!. \\
\end{aligned}
\end{align}
\end{subequations}

According to this replicator dynamics defined above, the number of user devices selecting ECP $n$ will increase when $\pi \left( n,\mathbf{x}\left( t \right),\mathbf{r}\left( t \right) \right)>\pi \left( \mathbf{x}\left( t \right),\mathbf{x}\left( t \right),\mathbf{r}\left( t \right) \right)$, and vice versa.

\subsection{Existence and Uniqueness of Equilibrium}

According to the population replicator dynamic formulated in (\ref{equ:population_dynamic}) - (\ref{equ:x_dynamic}) and the established hierarchical game framework, there exists interactions between decisions of computing service selection made by users and computing resource pricing/allocation controls of CPs. In other words, the evolution of population distribution state defined in (\ref{equ:population_dynamic}) is controlled by the pricing strategy of the CCP and resource requests of ECPs. Next, Theorem~\ref{theorem1} presents that under these controls, there exists the unique population distribution state $\mathbf{x}\left( t \right)$ that constitutes the solution of (\ref{equ:x_dynamic}).

\begin{theorem}
\label{theorem1}
Consider a dynamic service selection system with a fixed population. For the evolutionary behavior of users among different strategies defined as (\ref{equ:population_dynamic}) with initial condition $\mathbf{x}\left( 0 \right)={{\mathbf{x}}_{0}}$, if resource allocation vector $\mathbf{r}\left( t \right)$ is a vector of measurable functions on $\left[ 0,\infty  \right)$, then there exists the unique population distribution state $\mathbf{x}\left( t \right)$ constitute the solution of (\ref{equ:x_dynamic}) for all $t\in \left[ 0,\infty  \right)$.
\end{theorem}

\begin{proof}
Given population distribution state $\mathbf{x}\left( t \right)$ and computing resource allocation state $\mathbf{r}\left( t \right)$, let ${{f}_{n}}\left( \mathbf{x}\left( t \right),\mathbf{r}\left( t \right) \right)$ be the right side of (\ref{equ:population_dynamic}), i.e.,
\begin{eqnarray}
\begin{aligned}
  &{{f}_{n}}\left( \mathbf{x}\left( t \right),\mathbf{r}\left( t \right) \right) \\
 \triangleq & \delta {{x}_{n}}\left( t \right)\left[ \pi \left( n,\mathbf{x}\left( t \right),\mathbf{r}\left( t \right) \right)-\pi \left( \mathbf{x}\left( t \right),\mathbf{x}\left( t \right),\mathbf{r}\left( t \right) \right) \right].
\end{aligned}
\end{eqnarray}
Given a fixed $t$, the partial derivative of ${{f}_{n}}\left( \mathbf{x}\left( t \right),\mathbf{r}\left( t \right) \right)$ with respect to $\mathbf{x}\left( t \right)$ is continuous. In addition, if $\mathbf{r}\left( t \right)$ is measurable on $\left[ 0,\infty  \right)$, then ${{f}_{n}}\left( \mathbf{x}\left( t \right),\mathbf{r}\left( t \right) \right)$ is also measurable for fixed ${{x}_{n}}\left( t \right)$ on the same interval. Furthermore, when given any closed bounded set $\Delta \in \mathbb{X}$ and closed interval $\left[ a,b \right]\in \left[ 0,\infty  \right)$, there always exists a positive $I$ to construct an integrable function~\cite{zhu2014pricing, Friesz2007Dynamic, Ho1970Differential} by
\begin{eqnarray}
{{I}_{n}}\left( t \right)=\left| \frac{\delta \beta \left( {{R}_{n}}+{{R}_{c}}{{r}_{n}}\left( t \right) \right)}{K{{p}_{n}}} \right|+I\left| \Theta  \right|,
\end{eqnarray}
where
\begin{eqnarray}
\Theta \!=\!\frac{\delta \beta}{K} \left[ \sum\limits_{n=1}^{N}{\frac{{{R}_{n}}\!+\!{{R}_{c}}{{r}_{n}}\left( t \right)}{{{p}_{n}}}}\!+\!\frac{{{R}_{c}}\left( 1\!-\!\sum\nolimits_{n=1}^{N}{{{r}_{n}}\left( t \right)} \right)}{{{p}_{c}}} \right].
\label{equ:Theta}
\end{eqnarray}
Obviously, it holds that $\left| {{f}_{n}}\left( \mathbf{x}\left( t \right),\mathbf{r}\left( t \right) \right) \right|\le {{I}_{n}}\left( t \right)$ and $\left| {\partial {{f}_{n}}\left( \mathbf{x}\left( t \right),\mathbf{r}\left( t \right) \right)}/{\partial {{x}_{n}}\left( t \right)}\; \right|\le {{I}_{n}}\left( t \right)$, for all $\left( \mathbf{x},t \right)\in \Delta \times \left[ a,b \right]$. Therefore, we have $\left| {{f}_{n}}\left( {{\mathbf{x}}^{*}}\left( t \right),\mathbf{r}\left( t \right) \right)-{{f}_{n}}\left( \mathbf{x}\left( t \right),\mathbf{r}\left( t \right) \right) \right|=\Theta \left| {{\mathbf{x}}^{*}}\left( t \right)-\mathbf{x}\left( t \right) \right|$. Denote ${{\Theta }_{m}}=\max \left\{ \Theta  \right\}$, then we can further derive that
\begin{eqnarray}
\left| {{f}_{n}}\!\left( {{\mathbf{x}}^{*}}\left( t \right),\mathbf{r}\left( t \right) \right)\!-\!{{f}_{n}}\!\left( \mathbf{x}\left( t \right),\mathbf{r}\left( t \right) \right) \right|\!\le\! {{\Theta }_{m}}\left| {{\mathbf{x}}^{*}}\left( t \right)\!-\!\mathbf{x}\left( t \right) \right|,
\end{eqnarray}
which implies that ${{f}_{n}}\left( \mathbf{x}\left( t \right),\mathbf{r}\left( t \right) \right)$ satisfies the global Lipschitz condition. According to the analysis above, we can conclude that the solution to this dynamical population evolutionary system under controls of the CCP and ECPs is unique and exists globally~\cite{zhu2014pricing, jiang2013joint, 8939471, jiang2013renewal}.
This completes the proof of Theorem~\ref{theorem1}.

\end{proof}

\subsection{Analysis of Evolutionary Stable State (ESS)}

Consider a situation where a small proportion of user devices switching to a different mixed strategy $\mathbf{y}\ne \mathbf{x}$. Then these user devices can be regarded as mutants of the population. Denote the size of these mutants by a normalized value $\varepsilon \in \left( 0,1 \right)$. Then the population state after mutation can be given by $\left( 1-\varepsilon  \right)\mathbf{x}+\varepsilon \mathbf{y}$~\cite{zhu2014pricing}. According to definitions above, we first give the definition of \emph{Evolutionary Stable Strategy (ESS)} in Definition~\ref{definition:EG_ESS}.

\begin{definition}
\label{definition:EG_ESS}
(\textbf{Evolutionary Stable Strategy}) A strategy ${{\mathbf{x}}^{*}}$ is an ESS, if $\forall \mathbf{x}\ne {{\mathbf{x}}^{*}}$, there exist some ${{\varepsilon }_{\mathbf{x}}}\in \left( 0,1 \right)$~\footnote{${{\varepsilon }_{\mathbf{x}}}$ represents the maximum proportion of users selecting mutant strategies that can be resisted by the ESS. A large ${{\varepsilon }_{\mathbf{x}}}$ indicates that the ESS is robust.} such that $\forall \varepsilon \in \left( 0,{{\varepsilon }_{\mathbf{x}}} \right)$, the following inequality holds.
\begin{eqnarray}
\pi \left( {{\mathbf{x}}^{*}},\left( 1-\varepsilon  \right){{\mathbf{x}}^{*}}+\varepsilon \mathbf{x},\mathbf{r} \right)>\pi \left( \mathbf{x},\left( 1-\varepsilon  \right){{\mathbf{x}}^{*}}+\varepsilon \mathbf{x},\mathbf{r} \right),
\end{eqnarray}
where $\pi \!\left( {{\mathbf{x}}^{*}},\left( 1\!-\!\varepsilon  \right){{\mathbf{x}}^{*}}\!+\!\varepsilon \mathbf{x},\mathbf{r} \right)$ and $\pi \left( \mathbf{x},\left( 1\!-\!\varepsilon  \right){{\mathbf{x}}^{*}}\!+\!\varepsilon \mathbf{x},\mathbf{r} \right)$ are the expected utilities of non-mutants and mutants, respectively.
\end{definition}

Considering that the ESS is the best response to the evolutionary system, then an ESS is also a Nash Equilibrium (NE). In addition, the evolutionary stability of ESS provides a string refinement of the NE. Moreover, in the NE, a single user cannot benefit through deviating from the equilibrium strategy. On the contrary, the ESS can avoid the deviation behavior of a set of players.
Next, we summarize that the evolutionary service selection of users presents the globally asymptotical stability converging to the ESS in Theorem~\ref{theorem2}.

\begin{theorem}
\label{theorem2}
Consider a dynamic computing resource selection system with a fixed population. For the evolutionary behavior of the population among different strategies defined as (\ref{equ:population_dynamic}) with any initial condition $\mathbf{x}\left( 0 \right)={{\mathbf{x}}_{0}}$ ($x_n, x_c\in\left(0,1\right)$, $\forall n\in\mathsf{\mathcal{S}}$), the replicator dynamics for resource selection is globally asymptotically stable and converges to the ESS of game.
\end{theorem}

\begin{proof}
Considering $\mathbf{x}={{\left[ {{x}_{1}},{{x}_{2}},\cdots ,{{x}_{N}},{{x}_{c}} \right]}^{T}}$ and the replicator dynamic derived in (\ref{equ:x_dynamic}), we have $\mathbf{\dot{x}}=\bm{\Pi x}+{{\bm{\pi }}_{o}}$, where $\mathbf{\Pi }$ is a matrix with dimensional $\left( N+1 \right)\times \left( N+1 \right)$, and
\begin{eqnarray}
\begin{aligned}
  {{\bm{\pi }}_{o}}=& \left[ \frac{\delta \beta \left( {{R}_{1}}+{{R}_{c}}{{r}_{1}}\left( t \right) \right)}{K{{p}_{1}}},\cdots , \frac{\delta \beta \left( {{R}_{N}}+{{R}_{c}}{{r}_{N}}\left( t \right) \right)}{K{{p}_{N}}}, \right. \\
 & {{\left. \frac{\delta \beta {{R}_{c}}\left( 1-\sum\nolimits_{m=1}^{N}{{{r}_{m}}\left( t \right)} \right)}{K{{p}_{c}}} \right]}^{T}}.
\end{aligned}
\end{eqnarray}
Therefore, the characteristic function of (\ref{equ:x_dynamic}) can be derived as $\det \left( \gamma \mathbf{I}-\bm{\Pi } \right)={{\left( \gamma +\Theta  \right)}^{N+1}}=0$, where $\mathbf{I}$ is the identity matrix and $\Theta $ is determined by (\ref{equ:Theta}). In addition, (\ref{equ:Theta}) implies that $\Theta >0$ for all resource allocation states $\mathbf{r}\left(t\right)$. Consequently, $\bm{\Pi }$ always has $N\text{+1}$ negative eigenvalues, which means that the replicator dynamics for service selection is globally asymptotically stable and converges to the ESS of evolutionary game.
This completes the proof of Theorem~\ref{theorem2}.

\end{proof}
\vspace{-5mm}
\section{Stackelberg Differential Game Based Dynamic Computational Power Pricing and Allocation}
\label{stackelberg}

The CCP and ECPs need to make the optimal decisions on computational power pricing and the amount of computational power requests, respectively, considering
the dynamic service selection of user devices $\mathbf{x}\left( t \right)$. For the CCP, decreasing price $p\left( t \right)$ might incentivize ECPs to request and buy more remote computing resource, which will increase the sharing of CCP resource. However, the user devices accessing the CCP might then leave for other ECPs since that the amount of their received computational power decreases.
On the other hand, for ECPs, increasing the amount of computational power requests will improve the utilities obtained by user devices according to (\ref{equ:payoff_user}), which will attract more users' selection according to the replicator dynamics as (\ref{equ:population_dynamic}). Then further increasing number of users assessing will reduce the utility obtained by each user device.
To analyze this dynamic and interactive decision making problem and then facilitate the computing resource trading between hierarchical CPs,
we formulate a Stackelberg differential game, in which the CCP and ECPs perform as the leader and followers, respectively.
To search the optimal strategies, an open-loop Stackelberg equilibrium is analyzed as the solution of the game.
\vspace{-2mm}

\subsection{Formulation of Stackelberg Differential Game}

As shown in Fig.~\ref{fig:system_model}, the single CCP and $N$ ECPs perform as the players of the Stackelberg game. Specifically, the CCP, as the game leader, first announces its unit computing resource price $p\left( t \right)$, according to which ECPs, who are the followers of the game, then make their responding decisions of resource requests ${{r}_{n}}\left( t \right)$.
In this work, we assume that both the CCP and ECPs are rational so that they can make the best response to the system states and strategies of other players, and follow the strategies made by SDN controllers. In addition, in the established Stackelberg differential game, the CCP and ECPs are willing to optimize their integral utilities over the time horizon $\left[ 0,T \right]$, but not the current utilities, by dynamically controlling their pricing and request strategies, responsibility.

As the followers of Stackelberg game, ECPs can optimize the amount of requesting computational power when observing the unit price released by the CCP. However, the time-varying strategies of all ECPs during the time horizon cannot be observed by the CCP in the present moment. Therefore, this work considers that the CCP is able to learn and predict the expected best response of ECPs and then make its pricing strategy dynamically. Next, we formulate the maximization problems of integral utility for both the CCP and ECPs.

\subsubsection{Maximization of Integral Utility for ECPs}

Consider that the utility of each ECP is composed of economic profits and penalty of resource sharing performance. To be specific, by setting accessing price ${{p}_{n}}$ ($\forall n\in \mathsf{\mathcal{N}}$), ECP $n$ obtains the revenue from users selecting to it, which is depends on the number of subscribed users, i.e., $K{{p}_{n}}{{x}_{n}}\left( t \right)$. In addition, when requesting the CCP for proportion of could computational power ${{r}_{n}}\left( t \right)$, ECP $n$ will be charged ${{R}_{c}}p\left( t \right){{r}_{n}}\left( t \right)$ by the CCP. Moreover, the costs resulting from the mismatch between resource supply and demand are also taken account of by ECP $n$ when optimizing its resource request strategy ${{r}_{n}}\left( t \right)$, from the performance aspect. This mismatch can be modeled as the distance between the current computational power requirements from all subscribing user devices and the current total computing resource can be provided after receiving the CCP's sharing resource.
In the follower layer of Stackelberg game, each ECP is trying to maximize its own profits while minimize the costs resulting from the mismatch between resource supply and demand.
Consequently, the instantaneous utility of ECP $n$ can be given by
\begin{eqnarray}
\begin{aligned}
 {{u}_{n}}\!\left( {{r}_{n}}\!\left( t \right),\mathbf{x}\!\left( t \right),p\!\left( t \right) \right)\!& =\! {{\eta }_{1}}{{p}_{n}}N{{x}_{n}}\!\left( t \right)\!-\!{{\eta }_{2}}{{R}_{c}}p\!\left( t \right){{r}_{n}}\!\left( t \right) \\
 & -\!{{\eta }_{3}}{{\left[ K\varphi {{x}_{n}}\!\left( t \right)\!-\!\left( {{R}_{n}}\!+\!{{R}_{c}}{{r}_{n}}\!\left( t \right) \right) \right]}^{2}},
\end{aligned}  \label{equ:utility:fcp}
\end{eqnarray}
where $\varphi>0$ is defined as a nominal value of accessible computing rate for all user devices, and ${{\eta }_{1}}$, ${{\eta }_{2}}$ and ${{\eta }_{3}}$ are positive weight factors. In addition, the third term in (\ref{equ:utility:fcp}) reflects the matching between the computing resource requirement and available service capacity.

According to the instantaneous utility function (\ref{equ:utility:fcp}), the ECP utility depends not only on the received computational power shared by the CCP, but also the population distribution of user devices among different CPs. Therefore, given the pricing strategy of the CCP, the integral utility maximization problem for ECPs can be established as an optimal control problem subject to the population state of evolutionary game operated in the user-level, which is given by
\begin{subequations}
\label{equ:opt:FCP}
\begin{align}
    \begin{aligned} &\underset{{{r}_{n}}\left( t \right)}{\mathop{\max }}\,\\&~ \\&~ \\&~\end{aligned}  & ~~\begin{aligned}
  & U_{n}^{\operatorname{int}}\left( {{r}_{n}}\left( t \right),\mathbf{x}\left( t \right),p\left( t \right) \right)  \\
  & =\int_{0}^{T}{{{e}^{-\rho t}}\left[ {{\eta }_{1}}{{p}_{n}}N{{x}_{n}}\left( t \right)-{{\eta }_{2}}{{R}_{c}}p\left( t \right){{r}_{n}}\left( t \right) \right.}  \\
  & ~~~~~\left. -{{\eta }_{3}}{{\left[ K\varphi {{x}_{n}}\left( t \right)-\left( {{R}_{n}}+{{R}_{c}}{{r}_{n}}\left( t \right) \right) \right]}^{2}} \right]dt;\end{aligned} \label{equ:opt:FCP:obj} \\
  \begin{aligned}& \text{s}\text{.t}\text{.}\\&~  \end{aligned} & ~~\begin{aligned} {{{\dot{x}}}_{n}}\left( t \right)=&\delta {{x}_{n}}\left( t \right)\left[ \pi \left( n,\mathbf{x}\left( t \right),\mathbf{r}\left( t \right) \right) \right.  \\
  & ~~\left. -\pi \left( \mathbf{x}\left( t \right),\mathbf{x}\left( t \right),\mathbf{r}\left( t \right) \right) \right],\forall n\in \mathsf{\mathcal{N}},\end{aligned} \label{equ:opt:FCP:con1}  \\
  & ~~\begin{aligned} {{{\dot{x}}}_{c}}\left( t \right)=&\delta {{x}_{c}}\left( t \right)\left[ \pi \left( n,\mathbf{x}\left( t \right),\mathbf{r}\left( t \right) \right) \right.  \\
  & ~~\left. -\pi \left( \mathbf{x}\left( t \right),\mathbf{x}\left( t \right),\mathbf{r}\left( t \right) \right) \right],\end{aligned} \label{equ:opt:FCP:con2} \\
  & ~~~\mathbf{x}\left( 0 \right)={{\mathbf{x}}_{0}}, \label{equ:opt:FCP:con3} \\
  & ~~~{{r}_{n}}\left( t \right)\in \mathsf{\mathcal{R}},\forall n\in \mathsf{\mathcal{N}}. \label{equ:opt:FCP:con4}
\end{align}
\end{subequations}
In (\ref{equ:opt:FCP}), $\rho >0$ denotes the discount rate influencing the discount value of future utilities.

\subsubsection{Maximization of Integral Utility for CCP}

Similarly, the CCP optimizes its pricing strategy to maximize the profits paid by the subscribed users and the ECPs receiving the CCP's computational power, while minimize the costs resulting from the performance discrepancy. Then we have the instantaneous utility of CCP as follows,
\begin{eqnarray}
\begin{aligned}
 {{u}_{c}}\!\left( p\!\left( t \right),{{r}_{c}}\!\left( t \right),\mathbf{r}\!\left( t \right) \right) =\!{{\xi }_{1}}{{p}_{c}}N{{x}_{c}}\!\left( t \right)\!+\!{{\xi }_{2}}{{R}_{c}}\!\sum\limits_{n=1}^{N}{p\!\left( t \right){{r}_{n}}\!\left( t \right)} \\
  -{{\xi }_{3}}\!{{\left[ K\varphi {{x}_{c}}\!\left( t \right)\!-\!{{R}_{c}}\!\left( 1\!-\!\sum\limits_{n=1}^{N}{{{r}_{n}}\!\left( t \right)} \right) \right]}^{2}},
\end{aligned}  \label{equ:utility:ccp}
\end{eqnarray}
where ${{\xi }_{1}}$, ${{\xi }_{2}}$ and ${{\xi }_{3}}$ are positive weight factors. Therefore, the integral utility maximization problem for the CCP can be also established as an optimal control problem subject to the population state of evolutionary game operated in the user-level, which can be formulated as
\begin{subequations}
\label{equ:opt:CCP}
\begin{align}
    \begin{aligned} &\underset{{p}\left( t \right)}{\mathop{\max }}\,\\&~ \\&~ \\&~\end{aligned}  & ~~\begin{aligned}
  & U_{c}^{\operatorname{int}}\left( p\left( t \right),\mathbf{x}\left( t \right),\mathbf{r}\left( t \right) \right)  \\
  & =\!\int_{0}^{T}\!{{{e}^{-\rho t}}\!\left[ {{\xi }_{1}}{{p}_{c}}N{{x}_{c}}\!\left( t \right)\!+\!{{\xi }_{2}}{{R}_{c}}\!\sum\limits_{n=1}^{N}{p\left( t \right)\!{{r}_{n}}\!\left( t \right)} \right.}  \\
  & ~~~~~\left. -{{\eta }_{3}}{{\left[ K\varphi {{x}_{n}}\left( t \right)-\left( {{R}_{n}}+{{R}_{c}}{{r}_{n}}\left( t \right) \right) \right]}^{2}} \right]dt;\end{aligned} \label{equ:opt:CCP:obj} \\
  \begin{aligned}& \text{s}\text{.t}\text{.}\\&~  \end{aligned} & ~~\begin{aligned} {{{\dot{x}}}_{n}}\left( t \right)=&\delta {{x}_{n}}\left( t \right)\left[ \pi \left( n,\mathbf{x}\left( t \right),\mathbf{r}\left( t \right) \right) \right.  \\
  & ~~\left. -\pi \left( \mathbf{x}\left( t \right),\mathbf{x}\left( t \right),\mathbf{r}\left( t \right) \right) \right],\forall n\in \mathsf{\mathcal{N}},\end{aligned} \label{equ:opt:CCP:con1}  \\
  & ~~\begin{aligned} {{{\dot{x}}}_{c}}\left( t \right)=&\delta {{x}_{c}}\left( t \right)\left[ \pi \left( n,\mathbf{x}\left( t \right),\mathbf{r}\left( t \right) \right) \right.  \\
  & ~~\left. -\pi \left( \mathbf{x}\left( t \right),\mathbf{x}\left( t \right),\mathbf{r}\left( t \right) \right) \right],\end{aligned} \label{equ:opt:CCP:con2} \\
  & ~~~\mathbf{x}\left( 0 \right)={{\mathbf{x}}_{0}}, \label{equ:opt:CCP:con3} \\
  & ~~~{p}\left( t \right)\in \mathsf{\mathcal{R}}. \label{equ:opt:CCP:con4}
\end{align}
\end{subequations}

\subsection{Open-Loop Stackelberg Equilibrium Solutions}

In this part, we will analyze the open-loop solutions to the optimal computing resource pricing and requesting problems established above in (\ref{equ:opt:CCP}) and (\ref{equ:opt:FCP}) for the CCP and ECPs, respectively. For these optimization problems, if the CCP and ECPs choose to commit their strategies from outset, their information structure can be seen as an open-loop pattern, and their strategies become functions of the initial state ${{\mathbf{x}}_{0}}$, ${{\mathbf{r}}_{0}}$ and time $t$, for both the CCP and ECPs. Considering the Stackelberg differential game operation, it needs to search the optimal solution for each ECP first for the given CCP's pricing strategy, and then the CCP can make the decision on the computing price based on solutions of resource request strategies. Next, we will first analyze the optimal resource request problem for each ECP (follower) in a finite time period $\left[ 0,T \right]$, and then the optimal pricing strategy for the CCP (leader) will be obtained based on the ECPs' strategies.

In a Stackelberg differential game, an open-loop Stackelberg equilibrium is regarded as the optimal solution~\cite{1973On, 2001Existence}. So we first introduce the definitions of optimal control strategies for the CCP and ECPs.

\begin{definition}
\label{definition:OCS}
(\textbf{Optimal Control Strategy}) For the CCP, pricing strategy ${{p}^{*}}\left( t \right)$ is optimal if the following inequality holds for all feasible control paths $p\left( t \right)\ne {{p}^{*}}\left( t \right)$.
\begin{eqnarray}
U_{c}^{\operatorname{int}}\left( p^{*}\left( t \right),\mathbf{x}\left( t \right),\mathbf{r}\left( t \right) \right)\ge U_{c}^{\operatorname{int}}\left( p\left( t \right),\mathbf{x}\left( t \right),\mathbf{r}\left( t \right) \right). \label{equ:ocs:ccp}
\end{eqnarray}

Similarly, for ECP $n$ ($\forall n\in \mathsf{\mathcal{N}}$), the proportion of computational power request $r_{n}^{*}\left( t \right)$ is optimal if inequality (\ref{equ:ocs:fcp}) holds for all feasible control paths ${{r}_{n}}\left( t \right)\ne r_{n}^{*}\left( t \right)$.
\begin{eqnarray}
U_{n}^{\operatorname{int}}\left( r_{n}^{*}\left( t \right),\mathbf{x}\left( t \right),p\left( t \right) \right)\ge U_{n}^{\operatorname{int}}\left( {{r}_{n}}\left( t \right),\mathbf{x}\left( t \right),p\left( t \right) \right). \label{equ:ocs:fcp}
\end{eqnarray}
\end{definition}

Based on the definition of the optimal control strategy, we give the definition of open-loop Stackelberg game equilibrium in Definition~\ref{definition:SGE}.

\begin{definition}
\label{definition:SGE}
(\textbf{Open-loop Stackelberg Equilibrium}) Strategy profile ${{\Phi }^{*}}\left( t \right)\triangleq \left\{ {{p}^{*}}\left( t \right),{{\mathbf{r}}^{*}}\left( t \right) \right\}$ constitutes an open-loop Stackelberg equilibrium if ${{p}^{*}}\left( t \right)$ and ${{\mathbf{r}}^{*}}\left( t \right)$ are the optimal control strategies for the CCP and ECPs, respectively, given others' strategies.
\end{definition}

\subsubsection{Open-loop Stackelberg Equilibrium of ECPs}

In order to get equilibrium solutions of the optimization problem formulated in (\ref{equ:opt:FCP}), we need to establish the Hamiltonian system for each ECP. Then the open-loop equilibrium solutions of optimization problem can be characterized as the Pontryagin's Maximum Principle, which is the necessary conditions to find the candidate optimal strategies. First, we summarize the Pontryagin's Maximum Principle in Definition~\ref{definition:PMP}.

\begin{definition}
\label{definition:PMP}
(\textbf{Pontryagin's Maximum Principle for ECPs})
A set of controls $\left\{ r_{n}^{*}\left( t \right) \right\}$ constitutes an open-loop equilibrium to the optimization problem formulated in (\ref{equ:opt:FCP}), and ${{\mathbf{x}}_f^{*}}\left( t \right)$ is the corresponding population distribution state trajectory, if there exists a set of costate functions ${{\bm{\Lambda} }_{n}}\left( t \right)=\left[ \begin{array}{*{35}{l}}
   {{\lambda }_{n1}}\left( t \right) & {{\lambda }_{n2}}\left( t \right) & \cdots  & {{\lambda }_{nm}}\left( t \right) & \cdots  & {{\lambda }_{nN}}\left( t \right) & {{\lambda }_{nc}}\left( t \right)\\
\end{array} \right]$~\footnote{Considering that ${{r}_{c}}\left( t \right)=1-\sum\nolimits_{n=1}^{N}{{{r}_{n}}\left( t \right)}$, then element ${{\lambda }_{nc}}\left( t \right)$ in ${{\Lambda }_{n}}\left( t \right)$ can be eliminated.} such that the following relations are satisfied.
\begin{subequations}
\label{equ:PMP}
\begin{align}
  & r_{n}^{*}\!\left( t \right)\!=\!\arg \underset{{{r}_{n}}\!\left( t \right)}{\mathop{\max }}\,\!\left\{ {{u}_{n}}\!\left( {{r}_{n}}\!\left( t \right)\!,{{\mathbf{x}}^{*}}\!\left( t \right)\!,p\!\left( t \right) \right)\!+\!{{{\bm{\Lambda} } }_{n}}\!\left( t \right)\!{{{\mathbf{\dot{x}}}}^{*}}\!\left( t \right) \right\}\!, \label{equ:pmp1} \\
  & \dot{x}_{n}^{*}\!\left( t \right)\!=\!\delta x_{n}^{*}\!\left( t \right)\!\left[ \pi \left( n,{{\mathbf{x}}^{*}}\!\left( t \right),{{\mathbf{r}}^{*}}\!\left( t \right) \right)\!-\!\pi \left( {{\mathbf{x}}^{*}}\!\left( t \right),{{\mathbf{x}}^{*}}\!\left( t \right),{{\mathbf{r}}^{*}}\!\left( t \right) \right) \right], \label{equ:pmp2}\\
  & \dot{x}_{c}^{*}\!\left( t \right)\!=\!\delta x_{c}^{*}\!\left( t \right)\!\left[ \pi \left( c,{{\mathbf{x}}^{*}}\!\left( t \right),{{\mathbf{r}}^{*}}\!\left( t \right) \right)\!-\!\pi \left( {{\mathbf{x}}^{*}}\!\left( t \right),{{\mathbf{x}}^{*}}\!\left( t \right),{{\mathbf{r}}^{*}}\!\left( t \right) \right) \right], \label{equ:pmp5}\\
  &{{\mathbf{x}}^{*}}\left( 0 \right)=\mathbf{x}_{0}^{*},\label{equ:pmp3}\\
  &\begin{aligned}
 {{{\dot{\bm{\Lambda} }}}_{n}}\left( t \right) &=\rho {{\bm{\Lambda} }_{n}}\left( t \right) \\
 & -\frac{\partial \left[ {{u}_{n}}\left( {{r}_{n}}\!\left( t \right),\mathbf{x}\!\left( t \right),p\left( t \right) \right)\!+\!{{\bm{\Lambda}}_{n}}\!\left( t \right)\mathbf{\dot{x}}\!\left( t \right) \right]}{\partial {{\mathbf{x}}^{*}}\!\left( t \right)} .
\end{aligned}\label{equ:pmp4}
\end{align}
\end{subequations}
\end{definition}

In this system, the equilibrium solutions for ECPs are the solutions of the differential game. Therefore, these solutions are also constitute the Stackelberg equilibrium for ECPs. Then we first introduce the Hamiltonian system for each ECP as follows.
Based on the Pontryagin's Maximum Principle, the Hamiltonian system of ECP $n$ can be given by
\begin{eqnarray}
\begin{aligned}
  & {{H}_{n}}\left( {{r}_{n}}\left( t \right),p\left( t \right),\mathbf{x}\left( t \right),{{\bm{\Lambda}}_{n}}\left( t \right),t \right) \\
 \triangleq & {{u}_{n}}\left( {{r}_{n}}\left( t \right),\mathbf{x}\left( t \right),p\left( t \right) \right)+{{\bm{\Lambda}}_{n}}\left( t \right)\mathbf{\dot{x}}\left( t \right) ,
\end{aligned}\label{equ:hamiltonian}
\end{eqnarray}
where costate function ${{{\bm{\Lambda}}}_{n}}\left( t \right)$ is a function associate with population state $\mathbf{x}\left( t \right)$, and is defined by (\ref{equ:pmp4}). In addition, each element of costate function ${{\bm{\Lambda}}_{n}}\left( t \right)$, i.e., ${{\bm{\Lambda}}_{nm}}\left( t \right)$, is the costate variable of ECP $n$ associated with state ${{x}_{m}}$.
Based on the Hamiltonian function defined in (\ref{equ:hamiltonian}), the corresponding maximized Hamiltonian function is defined as follows:
\begin{eqnarray}
\begin{aligned}
  & H_{n}^{*}\left( \mathbf{x}\left( t \right),{{\bm{\Lambda}}_{n}}\left( t \right),t \right) \\
 \triangleq & \underset{{{r}_{n}}\left( t \right)}{\mathop{\max }}\,\left\{ {{H}_{n}}\!\left( {{r}_{n}}\!\left( t \right),p\left( t \right),\mathbf{x}\!\left( t \right),{{\bm{\Lambda}}_{n}}\!\left( t \right),t \right)\left| {{r}_{n}}\!\left( t \right)\!\in\! \mathsf{\mathcal{R}} \right. \right\} .
\end{aligned}
\end{eqnarray}

\begin{lemma}
\label{lemma1}
The optimal computational power rate solutions for ECP $n$ ($\forall n\in \mathsf{\mathcal{N}}$) is
\begin{eqnarray}
\begin{aligned}
  r_{n}^{*}\left( t \right)=-\frac{{{\eta }_{2}}}{2{{\eta }_{3}}{{R}_{c}}}p\left( t \right)+&\frac{K\varphi {{x}_{n}}\left( t \right)-{{R}_{n}}}{{{R}_{c}}} \\
  &~+\frac{1}{2{{\eta }_{3}}{{R}_{c}}}\frac{\delta \beta }{K}{{{\bm{\Lambda}}}_{n}}{{\mathbf{q}}_{n}}\left( \mathbf{x} \right),
\end{aligned}
\label{equ:allocation}
\end{eqnarray}
which also constitutes an open-loop Stackelberg equilibrium for ECP $n$. In (\ref{equ:allocation}), ${{\mathbf{q}}_{n}}\left( \mathbf{x} \right)$ is an $N$-dimension vector which is given by
\begin{eqnarray}
{{\mathbf{q}}_{n}}\left( \mathbf{x} \right)=\frac{1}{{{p}_{n}}}{{\mathbf{i}}_{n}}-\left( \frac{1}{{{p}_{n}}}-\frac{1}{{{p}_{c}}} \right){{\left[ \begin{matrix}
   {{x}_{1}} & {{x}_{2}} & \cdots  & {{x}_{N}}  \\
\end{matrix} \right]}^{T}},
\end{eqnarray}
where $N$-dimension vector ${{\mathbf{i}}_{n}}$ is a standard basis, i.e., its $n$-th element is $1$ and other elements are $0$.
\end{lemma}

\begin{proof}
According to the Pontryagin's Maximum Principle for ECPs, the optimal control strategy of optimization problem (\ref{equ:opt:FCP}) must also maximize the corresponding Hamiltonian function. Therefore, all candidates' optimal strategies have to satisfy the following necessary optimality conditions:
\begin{eqnarray}
\frac{\partial {{H}_{n}}\left( {{r}_{n}}\left( t \right),p\left( t \right),\mathbf{x}\left( t \right),{{\bm{\Lambda}}_{n}}\left( t \right),t \right)}{\partial {{r}_{n}}\left( t \right)}=0.
\label{equ:H:dev}
\end{eqnarray}

Then plug (\ref{equ:payoff_user}) into (\ref{equ:H:dev}), and the optimal computing resource request can be deduced as
\begin{small}
\begin{eqnarray}
\begin{aligned}
  r_{n}^{*}\left( t \right)=-\frac{{{\eta }_{2}}}{2{{\eta }_{3}}{{R}_{c}}}p\left( t \right)+&\frac{K\varphi {{x}_{n}}\left( t \right)-{{R}_{n}}}{{{R}_{c}}} \\
  &~+\frac{1}{2{{\eta }_{3}}{{R}_{c}}}\frac{\delta \beta }{K}{{\bm{\Lambda}}_{n}}{{\mathbf{q}}_{n}}\left( \mathbf{x} \right).
\end{aligned}
\end{eqnarray}
\end{small}

Furthermore, according to (\ref{equ:pmp4}), we can calculate all elements of ${{\bm{\Lambda}}_{n}}\left( t \right)$, which can be given by
\begin{subequations}
\label{equ:lambda_fcp}
\begin{align}
  & {{\dot{\lambda}}_{nm}}={{\lambda }_{nm}}\left( \rho +\Theta\left( \mathbf{r}\left( t \right) \right) \right),~m\ne n, \label{equ:lambda_fcp1} \\
  & {{\dot{\lambda}}_{nn}}={{\lambda }_{nn}}\left( \rho +\Theta\left( \mathbf{r}\left( t \right) \right) \right)-{{\eta }_{1}}{{p}_{n}}K,\label{equ:lambda_fcp2}
\end{align}
\end{subequations}
where $\Theta\left( \mathbf{r}\left( t \right) \right)$ is defined in (\ref{equ:Theta}).
This completes the proof of Lemma~\ref{lemma1}.
\end{proof}
\vspace{-2mm}

According to the optimal solutions summarized in Lemma~\ref{lemma1}, we can observe that the optimal computational power requests and allocation for ECPs is an decreasing function of pricing $p\left(t\right)$ determined by the CCP.

\subsubsection{Open-loop Stackelberg Equilibrium of CCP}

Similarly, we can obtain the open-loop equilibrium solutions of (\ref{equ:opt:CCP}) for the CCP based on the dynamic optimal control. In particular, with the definition of optimal strategy for the CCP as (\ref{equ:ocs:ccp}) in Definition~\ref{definition:OCS}, the open-loop equilibrium solutions for the CCP can be characterized as the Pontryagin's Maximum Principle for CCP, as summarized in following Definition~\ref{definition:PMP:ccp}.\vspace{-1mm}
\begin{definition}
\label{definition:PMP:ccp}
(\textbf{Pontryagin's Maximum Principle for CCP})
A set of controls $\left\{ {p}^{*}\left( t \right) \right\}$ constitutes an open-loop equilibrium to the optimization problem formulated in (\ref{equ:opt:CCP}), and ${{\mathbf{x}}^{*}}\left( t \right)$ is the corresponding population distribution state trajectory, if there exist costate functions $\mathbf{M} \left( t \right)={{\left[ \begin{matrix}
   {{\mu}_{c1}}\left( t \right) & {{\mu}_{c2}}\left( t \right) & \cdots  & {{\mu}_{cN}}\left( t \right)  \\
\end{matrix} \right]}^{T}}$ and ${\bm{\Psi}} \left( t \right)={{\left[ \begin{matrix}
   {{\bm{\Psi}}_{1}}\left( t \right) & {{\bm{\Psi}}_{2}}\left( t \right) & \cdots  & {{\bm{\Psi}}_{N}}\left( t \right)  \\
\end{matrix} \right]}^{T}}$ such that the following relations are satisfied.
\begin{small}
\begin{subequations}
\label{equ:PMP:ccp}
\begin{align}
  & {{p }^{*}}\!\left( t \right)\!=\!\arg \underset{\rho\! \left( t \right)}{\mathop{\max }}\,\!\left\{ {{H}_{c}}\left( p\left( t \right)\!,\mathbf{x}\!\left( t \right)\!,\mathbf{r}\!\left( t \right),{\bm{\Lambda}} \!\left( t \right),\mathbf{M} \!\left( t \right)\!,{\bm{\Psi}} \!\left( t \right) \right) \right\}, \label{equ:pmp:ccp1} \\
  & \dot{x}_{n}^{*}\!\left( t \right)\!=\!\delta x_{n}^{*}\!\left( t \right)\!\left[ \pi\! \left( n,{{\mathbf{x}}^{*}}\!\left( t \right),{{\mathbf{r}}^{*}}\!\left( t \right) \right)\!-\!\pi \!\left( {{\mathbf{x}}^{*}}\!\left( t \right),{{\mathbf{x}}^{*}}\!\left( t \right),{{\mathbf{r}}^{*}}\!\left( t \right) \right) \right], \label{equ:ccp:pmp2}\\
  & \dot{x}_{c}^{*}\!\left( t \right)\!=\!\delta x_{c}^{*}\!\left( t \right)\!\left[ \pi\! \left( c,{{\mathbf{x}}^{*}}\!\left( t \right),{{\mathbf{r}}^{*}}\!\left( t \right) \right)\!-\!\pi\! \left( {{\mathbf{x}}^{*}}\!\left( t \right),{{\mathbf{x}}^{*}}\!\left( t \right),{{\mathbf{r}}^{*}}\!\left( t \right) \right) \right], \label{equ:ccp:pmp3}\\
  &{{\mathbf{x}}^{*}}\left( 0 \right)=\mathbf{x}_{0}^{*},\label{equ:ccp:pmp4}\\
  &\begin{aligned}
 \dot{\mathbf{M} }\left( t \right)& =\rho \mathbf{M} \left( t \right) \\
 & -\frac{\partial {{H}_{c}}\left( p\left( t \right),\mathbf{x}\!\left( t \right),\mathbf{r}\!\left( t \right),{\bm{\Lambda}} \!\left( t \right),\mathbf{M} \!\left( t \right),{\bm{\Psi}}\! \left( t \right) \right)}{\partial {{\mathbf{x}}^{*}}\left( t \right)},
\end{aligned}\label{equ:ccp:pmp5}\\
&\begin{aligned}
 {{{\dot{\bm{\Psi} }}}_{n}}\!\left( t \right)&=\rho {{\bm{\Psi} }_{n}}\!\left( t \right) \\
 & -\frac{\partial {{H}_{c}}\left( p\left( t \right),\mathbf{x}\!\left( t \right),\mathbf{r}\!\left( t \right),{\bm{\Lambda}}\!\left( t \right),\mathbf{M}\! \left( t \right),{\bm{\Psi}} \!\left( t \right) \right)}{\partial {{\bm{\Lambda} }_{n}}\left( t \right)},
\end{aligned}\label{equ:ccp:pmp6}
\end{align}
\end{subequations}
\end{small}
where the Hamiltonian function of the CCP is given by
\begin{small}
\begin{eqnarray}
\label{equ:H:ccp}
\begin{aligned}
  & ~~~~{{H}_{c}}\left( p\left( t \right),\mathbf{x}\left( t \right),\mathbf{r}\left( t \right),{\bm{\Lambda}} \left( t \right),\mathbf{M} \left( t \right),{\bm{\Psi}} \left( t \right) \right) \\
 & ={{\xi }_{1}}{{p}_{c}}K{{x}_{c}}\left( t \right)+{{\xi }_{2}}{{R}_{c}}p\left( t \right)\sum\limits_{n=1}^{N}{{{r}_{n}}\left( t \right)} \\
 & -{{\xi }_{3}}{{\left[ K\varphi {{x}_{c}}\left( t \right)-{{R}_{c}}\left( 1-\sum\limits_{n=1}^{N}{{{r}_{n}}\left( t \right)} \right) \right]}^{2}} \\
 & +\sum\limits_{n=1}^{N}{{{\mu}_{cn} }\left( t \right){{{\dot{x}}}_{n}}\left( t \right)}+\sum\limits_{n=1}^{N}{\left( \sum\limits_{m=1}^{N}{{{\theta }_{nm}}\left( t \right){{{\dot{\lambda }}}_{nm}}\left( t \right)} \right)},
\end{aligned}
\end{eqnarray}
\end{small}
${\bm{\Lambda}} \left( t \right)={{\left[ \begin{matrix}
   {{\bm{\Lambda}}_{1}}\left( t \right) & {{\bm{\Lambda}}_{2}}\left( t \right) & \cdots  & {{\bm{\Lambda}}_{N}}\left( t \right)  \\
\end{matrix} \right]}^{T}}$ determined  by (\ref{equ:pmp4}), $\mathbf{M} \left( t \right)={{\left[ \begin{matrix}
   {{\mu}_{c1}}\left( t \right) & {{\mu}_{c2}}\left( t \right) & \cdots  & {{\mu}_{cN}}\left( t \right)  \\
\end{matrix} \right]}^{T}}$ and ${\bm{\Psi}} \left( t \right)={{\left[ \begin{matrix}
   {{\bm{\Psi}}_{1}}\left( t \right) & {{\bm{\Psi}}_{2}}\left( t \right) & \cdots  & {{\bm{\Psi}}_{N}}\left( t \right)  \\
\end{matrix} \right]}^{T}}$, where ${{\bm{\Psi}}_{n}}\left( t \right)={{\left[ \begin{array}{*{35}{l}}
   {{\theta }_{n1}}\left( t \right) & {{\theta }_{n2}}\left( t \right) & \cdots  & {{\theta }_{nm}}\left( t \right) & \cdots  & {{\theta }_{nN}}\left( t \right)  \\
\end{array} \right]}^{T}}$, are costate functions for the CCP.
\end{definition}

By solving optimization problem (\ref{equ:pmp:ccp1}) based on Hamiltonian function (\ref{equ:H:ccp}), we provide the optimal pricing strategy in Lemma~\ref{lemma2}
\begin{lemma}
\label{lemma2}
The optimal computational power pricing solutions for the CCP is
\begin{eqnarray}
{{p}^{*}}\left( t \right)\triangleq {{f}_{p}}\left( \mathbf{x}\left( t \right),\bm{\Lambda} \left( t \right),\mathbf{M}\left( t \right),\bm{\Psi} \left( t \right),t \right),\label{equ:pricing:f}
\end{eqnarray}
where
\begin{eqnarray}
\label{equ:pricing}
\begin{aligned}
  & ~~~~{{f}_{p}}\left( \mathbf{x}\left( t \right),\bm{\Lambda} \left( t \right),\mathbf{M}\left( t \right),\bm{\Lambda} \left( t \right),t \right) \\
 & \triangleq \frac{1}{2NB\left( {{\xi }_{2}}+{{\xi }_{3}}{{R}_{c}}NB \right)}\left\{ {{\xi }_{2}}\sum\limits_{n=1}^{N}{{{A}_{n}}} \right. \\
 & +2{{\xi }_{3}}NB\left[ K\varphi \left( 1-\sum\limits_{n=1}^{N}{{{x}_{n}}} \right)-{{R}_{c}}\left( 1-\sum\limits_{n=1}^{N}{{{A}_{n}}} \right) \right] \\
 & +\frac{\delta \beta B}{K}\sum\limits_{n=1}^{N}{{{\mu }_{cn}}\left[ -\frac{1}{{{p}_{n}}}-{{x}_{n}}\left( -\sum\limits_{n=1}^{N}{\frac{1}{{{p}_{n}}}}+\frac{N}{{{p}_{c}}} \right) \right]} \\
 & \left. +\frac{\delta \beta B}{K}\sum\limits_{n=1}^{N}{\sum\limits_{m=1}^{N}{{{\theta }_{nm}}{{\lambda }_{nm}}\left( -\sum\limits_{n=1}^{N}{\frac{1}{{{p}_{n}}}}+\frac{N}{{{p}_{c}}} \right)}} \right\} .
\end{aligned}
\end{eqnarray}
This optimal pricing ${{p}^{*}}\left( t \right)$ also constitutes an open-loop Stackelberg equilibrium for the CCP.
\end{lemma}

\begin{proof}
According to Lemma~\ref{lemma1}, optimal response $r_{n}^{*}\left( t \right)$ of ECP $n$ can be expressed as
\begin{eqnarray}
r_{n}^{*}\left( t \right)={{A}_{n}}\left( \mathbf{x}\left( t \right) \right)-Bp\left( x \right), \label{equ:r_simp}
\end{eqnarray}
where
\begin{subequations}
\begin{align}
&\begin{aligned}
  & {{A}_{n}}\left( \mathbf{x}\left( t \right),{{\bm{\Lambda} }_{n}}\left( t \right) \right) \\
 & ~~=\frac{K\varphi {{x}_{n}}\left( t \right)-{{R}_{n}}}{{{R}_{c}}}+\frac{1}{2{{\eta }_{3}}{{R}_{c}}}\frac{\delta \beta }{K}{{\bm{\Lambda} }_{n}}\left( t \right){{\mathbf{q}}_{n}}\left( \mathbf{x} \right),
\end{aligned}\\
& B=\frac{{{\eta }_{2}}}{2{{\eta }_{3}}{{R}_{c}}}.
\end{align}
\end{subequations}

As assumed previously, the CCP can learn and predict the optimal response $r_{n}^{*}\left( t \right)$ of ECP $n$, $\forall n\in \mathsf{\mathcal{N}}$.
Therefore, plugging (\ref{equ:r_simp}) into the Hamiltonian function of the CCP (\ref{equ:H:ccp}), then the Hamiltonian function of the CCP become a concave function with respect to $p\left( t \right)$. Thus the optimal pricing strategy ${{p}^{*}}\left( t \right)$ is unique for the CCP, which has to satisfy the following necessary optimality conditions
\begin{eqnarray}
\begin{aligned}
 & \frac{\partial {{H}_{c}}\left( p\left( t \right),\mathbf{x}\left( t \right),{{\mathbf{r}}^{*}}\left( t \right),\bm{\Lambda} \left( t \right),\mathbf{M} \left( t \right),\bm{\Psi} \left( t \right),t \right)}{\partial p\left( t \right)} \\
\triangleq & \frac{\partial {{H}_{c}}\left( p\left( t \right),\mathbf{x}\left( t \right),\bm{\Lambda }\left( t \right),\mathbf{M}\left( t \right),\bm{\Psi} \left( t \right),t \right)}{\partial p\left( t \right)} =0.
\end{aligned}
\end{eqnarray}

Taking the first derivative of ${{H}_{c}}\left(t \right)$ with respect to $p\left( t \right)$ and then we have
\begin{eqnarray}
\begin{aligned}
  & \left( 2N{{\xi }_{2}}B+2{{\xi }_{3}}{{R}_{c}}{{N}^{2}}{{B}^{2}} \right){{p}^{*}}\left( t \right) \\
 =& {{\xi }_{2}}\sum\limits_{n=1}^{N}{{{A}_{n}}}\\
 +&2{{\xi }_{3}}NB\left[ K\varphi \left( 1-\sum\limits_{n=1}^{N}{{{x}_{n}}} \right)-{{R}_{c}}\left( 1-\sum\limits_{n=1}^{N}{{{A}_{n}}} \right) \right] \\
 +&\frac{\delta \beta B}{K}\sum\limits_{n=1}^{N}{{{\mu }_{cn}}\left[ -\frac{1}{{{p}_{n}}}-{{x}_{n}}\left( -\sum\limits_{n=1}^{N}{\frac{1}{{{p}_{n}}}}+\frac{N}{{{p}_{c}}} \right) \right]} \\
 +&\frac{\delta \beta B}{K}\sum\limits_{n=1}^{N}{\sum\limits_{m=1}^{N}{{{\theta }_{nm}}{{\lambda }_{nm}}\left( -\sum\limits_{n=1}^{N}{\frac{1}{{{p}_{n}}}}+\frac{N}{{{p}_{c}}} \right)}}.
\end{aligned}
\end{eqnarray}
Therefore, the optimal pricing strategy denoted by (\ref{equ:pricing:f}) and (\ref{equ:pricing}) can be obtained.

Furthermore, according to (\ref{equ:ccp:pmp5}) and (\ref{equ:ccp:pmp6}), we can calculate all elements of $\mathbf{M}\left( t \right)$ and $\bm{\Psi}\left( t \right)$, and then we obtain
\begin{subequations}
\label{equ:lambda_ccp}
\begin{align}
  & {{\dot{\mu }}_{cn}}={{\mu }_{cn}}\left( \rho +\Theta\left( \mathbf{r}\left( t \right) \right) \right)-{{\xi }_{1}}{{p}_{c}}K,~\forall n\in \mathsf{\mathcal{N}}, \label{equ:mu_ccp1} \\
  & {{\dot{\theta }}_{nm}}={{\theta }_{nm}}\Theta\left( \mathbf{r}\left( t \right) \right),~\forall n,m\in \mathsf{\mathcal{N}},\label{equ:thita_ccp2}
\end{align}
\end{subequations}
where $\Theta\left( \mathbf{r}\left( t \right) \right)$ is defined through (\ref{equ:Theta}).
This completes the proof of Lemma~\ref{lemma2}.

\end{proof}

\subsubsection{Open-loop Stackelberg Equilibrium Solutions}

According to the optimal resource pricing and allocation strategies described in Lemma~\ref{lemma1} and Lemma~\ref{lemma2}, ${{p}^{*}}\left( t \right)$ and $r_{n}^{*}\left( n \right)$ ($\forall n\in \mathsf{\mathcal{N}}$) can be denoted by
\begin{subequations}
\label{equ:optimal:all}
\begin{align}
  & {{p}^{*}}\left( t \right)\triangleq {{f}_{p}}\left( \mathbf{x}\left( t \right),\bm{\Lambda} \left( t \right),\mathbf{M}\left( t \right),\bm{\Psi} \left( t \right),t \right) ,\\
 & r_{n}^{*}\left( t \right)\triangleq {{f}_{r}}\left( \mathbf{x}\left( t \right),\bm{\Lambda} \left( t \right),\mathbf{M}\left( t \right),\bm{\Psi} \left( t \right),t \right), \\
 & r_{c}^{*}\left( t \right)=1-\sum\nolimits_{n=1}^{N}{r_{n}^{*}\left( t \right)} .
\end{align}
\end{subequations}

Then substituting (\ref{equ:optimal:all}) into (\ref{equ:payoff_user}), (\ref{equ:lambda_fcp}), (\ref{equ:mu_ccp1}) and (\ref{equ:thita_ccp2}), and a dynamic control system composed of population distribution state $\mathbf{x}\left( t \right)$ and all costate variables ${{\bm{\Lambda} }_{n}}\left( t \right)$, $\mathbf{M}\left( t \right)$ and $\bm{\Psi }\left( t \right)$ can be provided as follows.
\begin{subequations}
\begin{align}
  & {{\mathbf{x}}^{*}}\left( t \right)\triangleq {{f}_{x}}\left( \mathbf{x}\left( t \right),\bm{\Lambda} \left( t \right),\mathbf{M}\left( t \right),\bm{\Psi} \left( t \right),t \right) ,\\
 & \bm{\Lambda} _{n}^{*}\left( t \right)\triangleq {{f}_{\Lambda ,n}}\left( \mathbf{x}\left( t \right),\bm{\Lambda} \left( t \right),\mathbf{M}\left( t \right),\bm{\Psi} \left( t \right),t \right), \\
 & {{\mathbf{M}}^{*}}\left( t \right)\triangleq {{f}_{M}}\left( \mathbf{x}\left( t \right),\bm{\Lambda} \left( t \right),\mathbf{M}\left( t \right),\bm{\Psi} \left( t \right),t \right), \\
 & {{\bm{\Psi} }^{*}}\left( t \right)\triangleq {{f}_{\Psi }}\left( \mathbf{x}\left( t \right),\bm{\Lambda} \left( t \right),\mathbf{M}\left( t \right),\bm{\Psi} \left( t \right),t \right).
\end{align}
\end{subequations}

The dynamic control system formulated above is a typical two-point boundary value problem (TPBVP)~\cite{anderson1981comparison, pachter2019toward}. By solving this problem, optimal controls ${{\mathbf{x}}^{*}}\left( t \right)$, $\bm{\Lambda }_{n}^{*}\left( t \right)$, ${{\mathbf{M}}^{*}}\left( t \right)$ and ${{\bm{\Psi} }^{*}}\left( t \right)$ can be obtained. Based on the optimal solutions of TPBVP, the open-loop Stackelberg game equilibrium ${{\Phi }^{*}}\left( t \right)\triangleq \left\{ {{p}^{*}}\left( t \right),{{\mathbf{r}}^{*}}\left( t \right) \right\}$ can be further derived.

\section{Simulation Results}
\label{simulation}

In this part, we will analyze the service selection behavior based on the evolutionary game, and then use MATLAB2019b to evaluate the performance of proposed computing resource pricing and allocation mechanisms based on the Stackelberg differential game.
First of all, we introduce the scenario setup of the simulations.
In the following simulations, we assume a typical ECC system, in which there are a single CCP and multiple ECPs who can access the computing resource of the CCP. These CPs provide edge and cloud computing services to $K=100$ user devices randomly distributed within the coverage of an ECC system.

\subsection{Evolution of Population Distribution}
\label{simu1}

For the numerical analysis, we first consider the situation of two ECPs, i.e., ECP 1 and ECP 2. The local available computational power of the two ECPs are set as ${{R}_{1}}=2kH/s$ and ${{R}_{2}}=1kH/s$, and the fixed access prices of the two ECPs are given by ${{p}_{1}}=0.3$ and ${{p}_{2}}=0.2$, respectively~\cite{zhu2014pricing}.
In addition, the initial population distribution state is set as ${{\mathbf{x}}_{0}}=\left[ {{x}_{1}}\left( t \right),{{x}_{2}}\left( t \right),{{x}_{c}}\left( t \right) \right]=\left[ 0.3,0.3,0.4 \right]$, and the initial computing resource request state of ECPs is set as ${{\mathbf{r}}_{0}}=\left[ {{r}_{1}}\left( 0 \right),{{r}_{2}}\left( 0 \right) \right]=\left[ 0,0 \right]$, which means that each ECP serves its users with its own computing resource at the beginning of the time horizon.
Consider different sharable cloud computational power of the CCP, i.e., ${{R}_{c}}=5kH/s$ indicating a service quality with high-computational power and ${{R}_{c}}=2kH/s$ indicating a service quality with low-computational power. Then the CCP fixes its access price ${{p}_{c}}$ by selecting values in $\left\{ 0.5>\max \left\{ {{p}_{1}},{{p}_{2}} \right\},0.2=\min \left\{ {{p}_{1}},{{p}_{2}} \right\} \right\}$, which can reflect different cost performance of CCP for users. Moreover, set the learning rate of users as $\delta =1$.

\begin{figure}[!t]
\begin{center}
\subfigure[${R}_{c} =5$]{\includegraphics[width=0.479\textwidth]{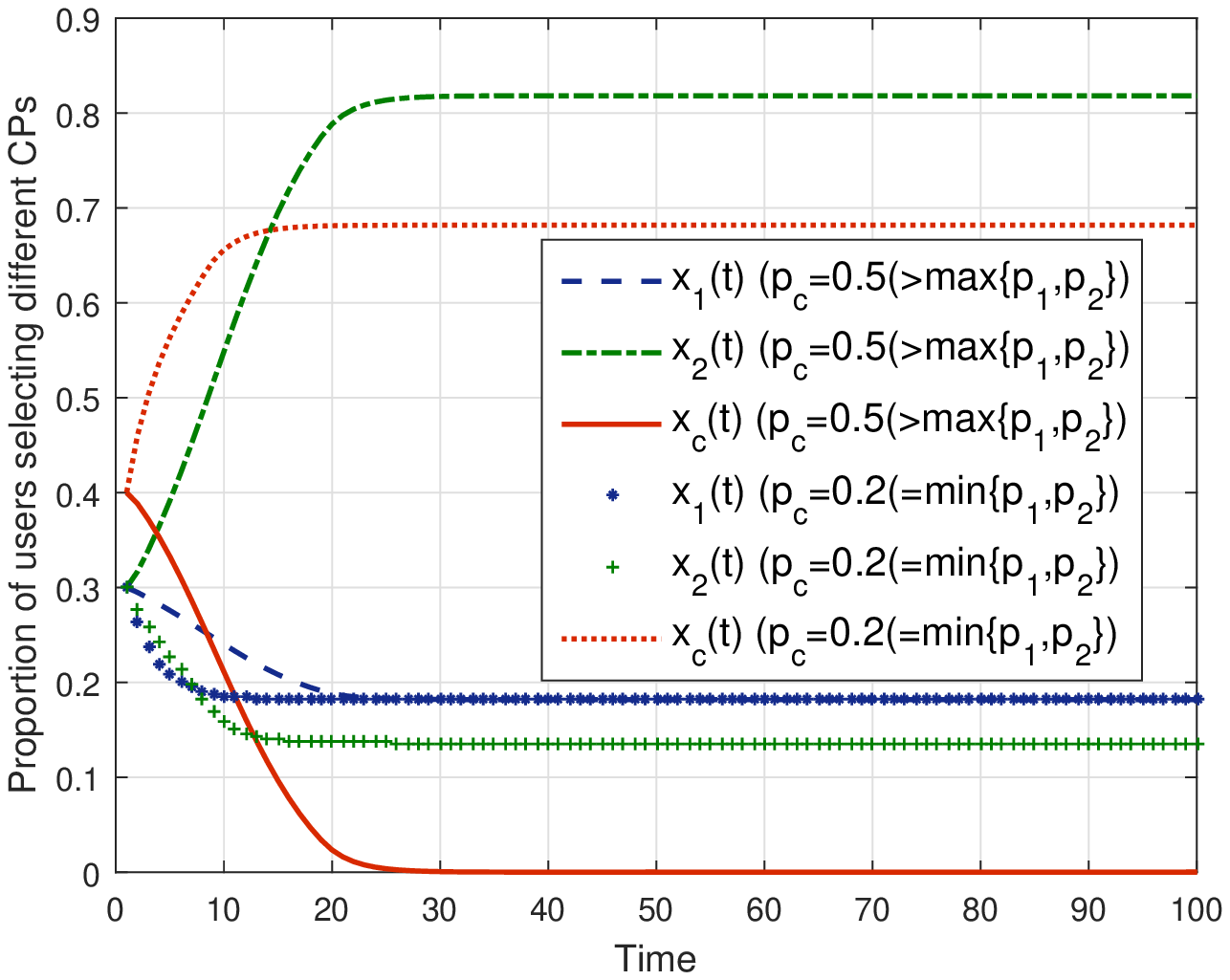}\label{fig:x_Rc5}}
\subfigure[${R}_{c} =2$]{\includegraphics[width=0.479\textwidth]{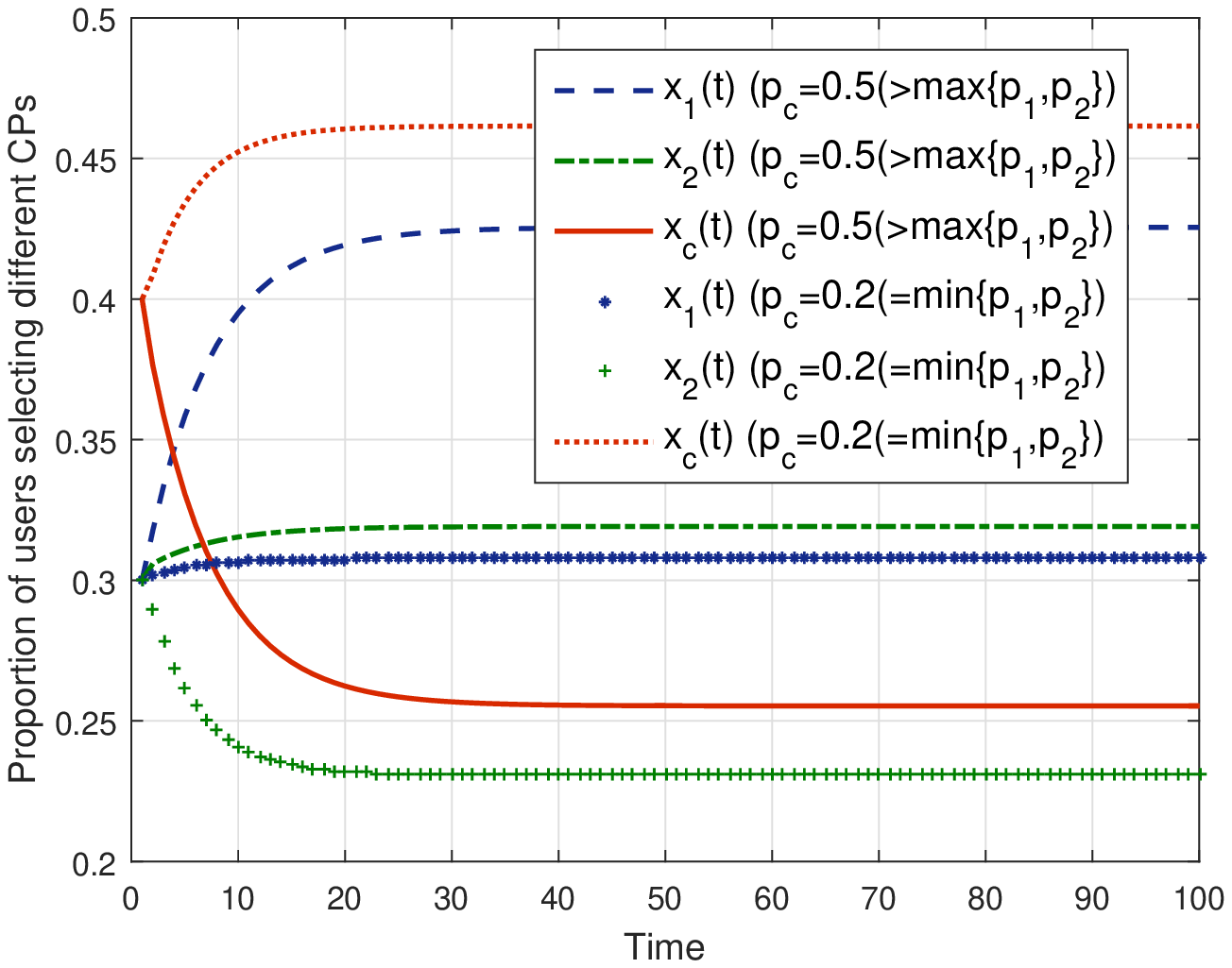}\label{fig:x_Rc2}}
\caption{Population distribution state evolution in the ECC system versus different $R_c$ and $p_c$.}
\label{fig:x}
\end{center}
\end{figure}

Then we first investigate the dynamics of population distribution state and the evolution process of service selection from initial state ${{\mathbf{x}}_{0}}$, which is subject to the control of resource pricing and allocation strategies.
Considering that the dynamic change of population distribution indicates the service selection adaptation of users, we record the population distribution state $\mathbf{x}\left( t \right)=\left[ {{x}_{1}}\left( t \right),{{x}_{2}}\left( t \right),{{x}_{c}}\left( t \right) \right]$ over time, and the results of which are shown in Fig.~\ref{fig:x}. Results in Fig.~\ref{fig:x} validate that the proportion of users selecting every CP converges to an equilibrium state at which there is no user willing to change its service selection strategy.

Then we analyze the influence of cloud computing capacity on user selection. As presented in Fig~\ref{fig:x_Rc5}, when $R_c=5kH/s$, the CCP setting a lower access price (${{p}_{c}}=0.2$) tends to attract more users to select its cloud resource directly, meanwhile share less computing resource to ECPs, although it possesses a larger computing capacity. On the contrary, when setting a rather high access price (${{p}_{c}}=0.5$), the CCP will share all of its computing resource to ECPs and then drive its subscribed users away to ECPs. In this case, the utility of CCP mainly comes from its sharing resource to ECPs. Then we analyze the resource selection evolution trajectory when CCP is limited with computational power, i.e., ${{R}_{c}}=2kH/s$, which are shown in Fig~\ref{fig:x_Rc2}. In this case, one can notice that the proportion of users selecting the CCP at equilibrium when ${{p}_{c}}=0.2$ is larger than that when ${{p}_{c}}=0.5$, which reflects the fact that the lower price will attract more users. Moreover, results in Fig~\ref{fig:x_Rc2} also imply that when the CCP has limited computing resource, the optimal pricing strategy for the CCP is reserving its resource to serve users directly by setting relatively higher unit price $p\left( t \right)$.
These results can reveal the interaction and influence between the evolutionary game in the user layer and the Stackelberg game in computing resource layer, and validate the rational user behaviors and market rules.

\begin{figure}[!t]
\begin{center}
\subfigure[Population distribution $x_n$, $x_c$]{\includegraphics[width=0.479\textwidth]{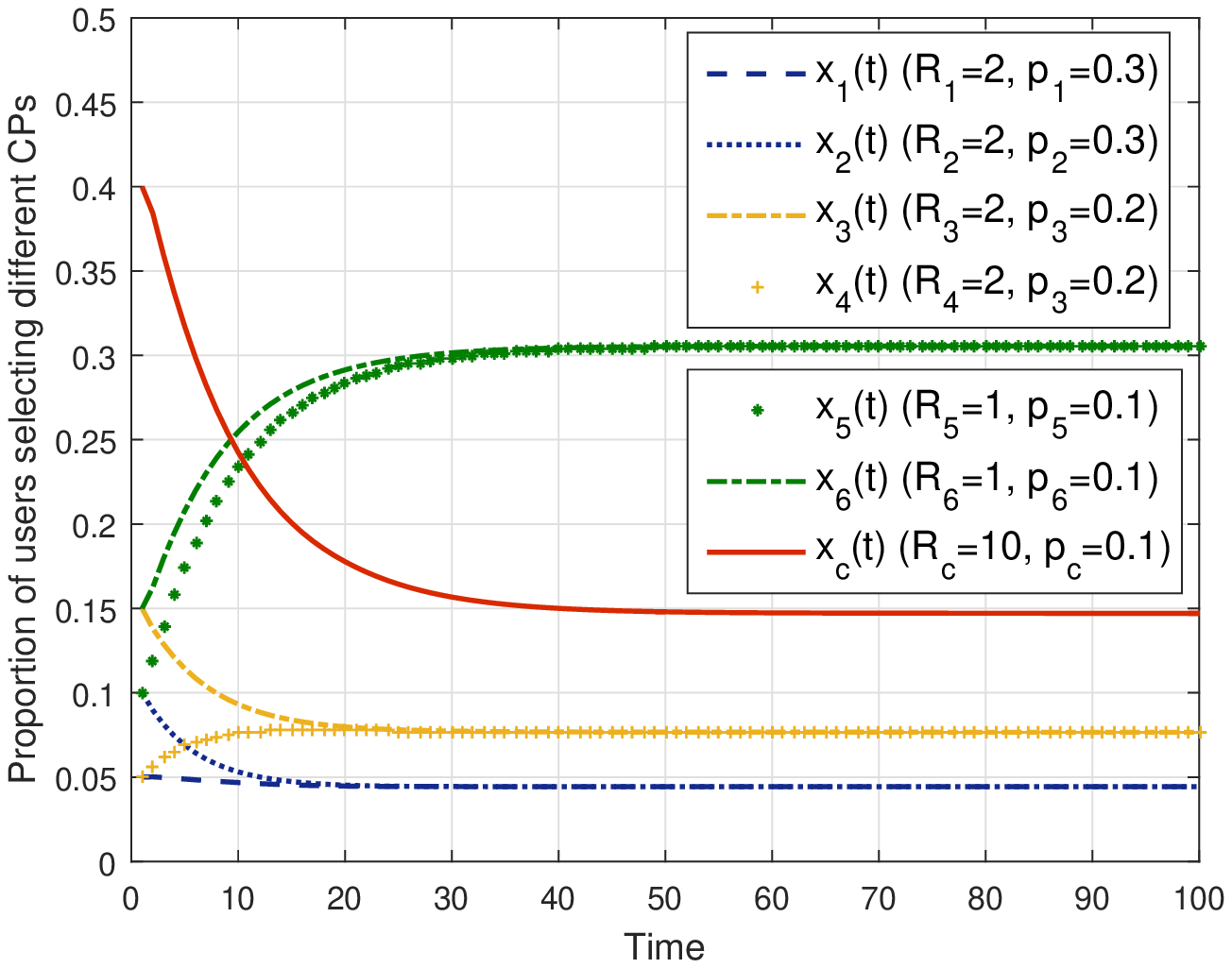} \label{fig:N_x}}
\subfigure[Resource allocation $r_n$, $r_c$]{\includegraphics[width=0.479\textwidth]{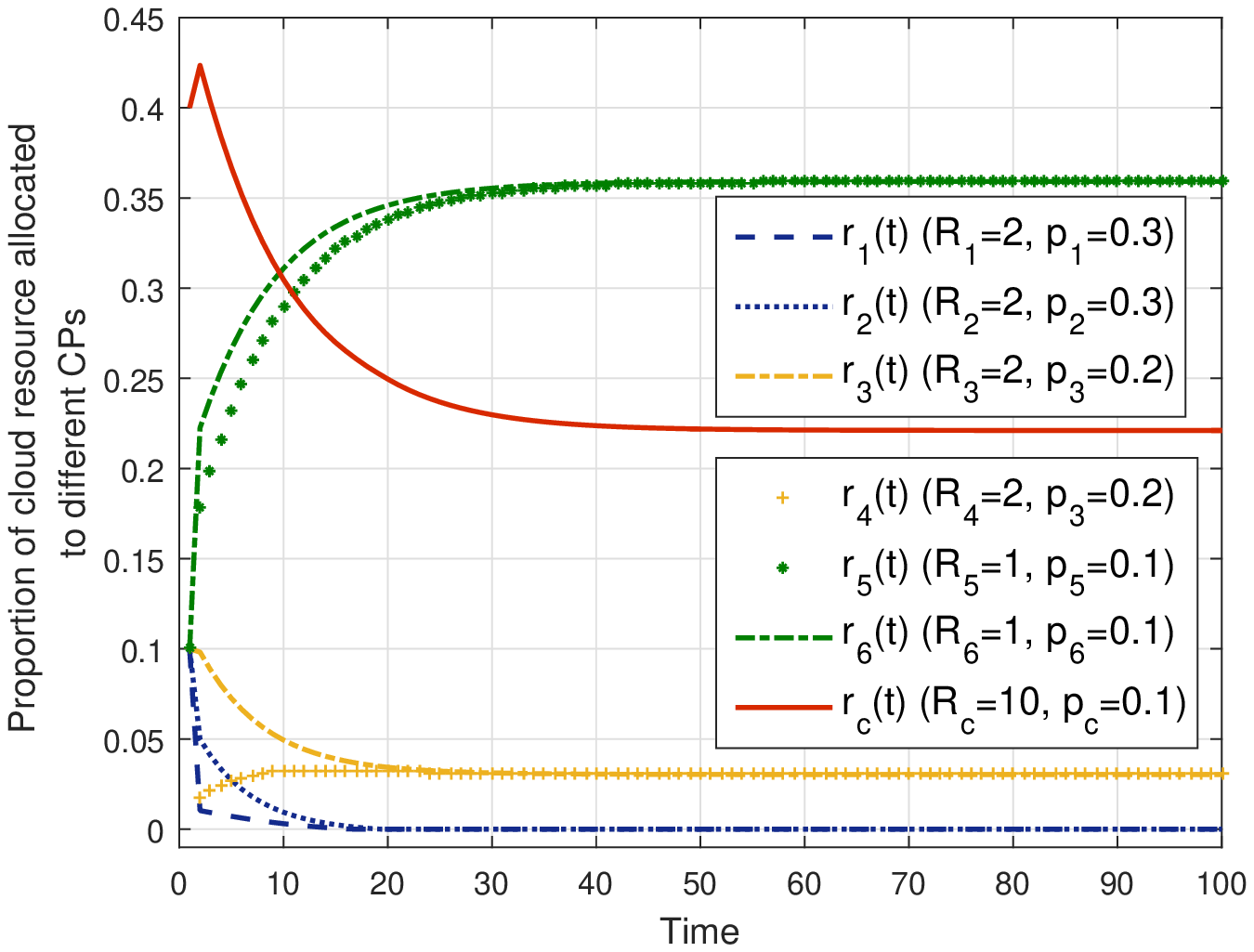}\label{fig:N_r}}
\subfigure[User utilities $\pi_n$, $\pi_c$]{\includegraphics[width=0.479\textwidth]{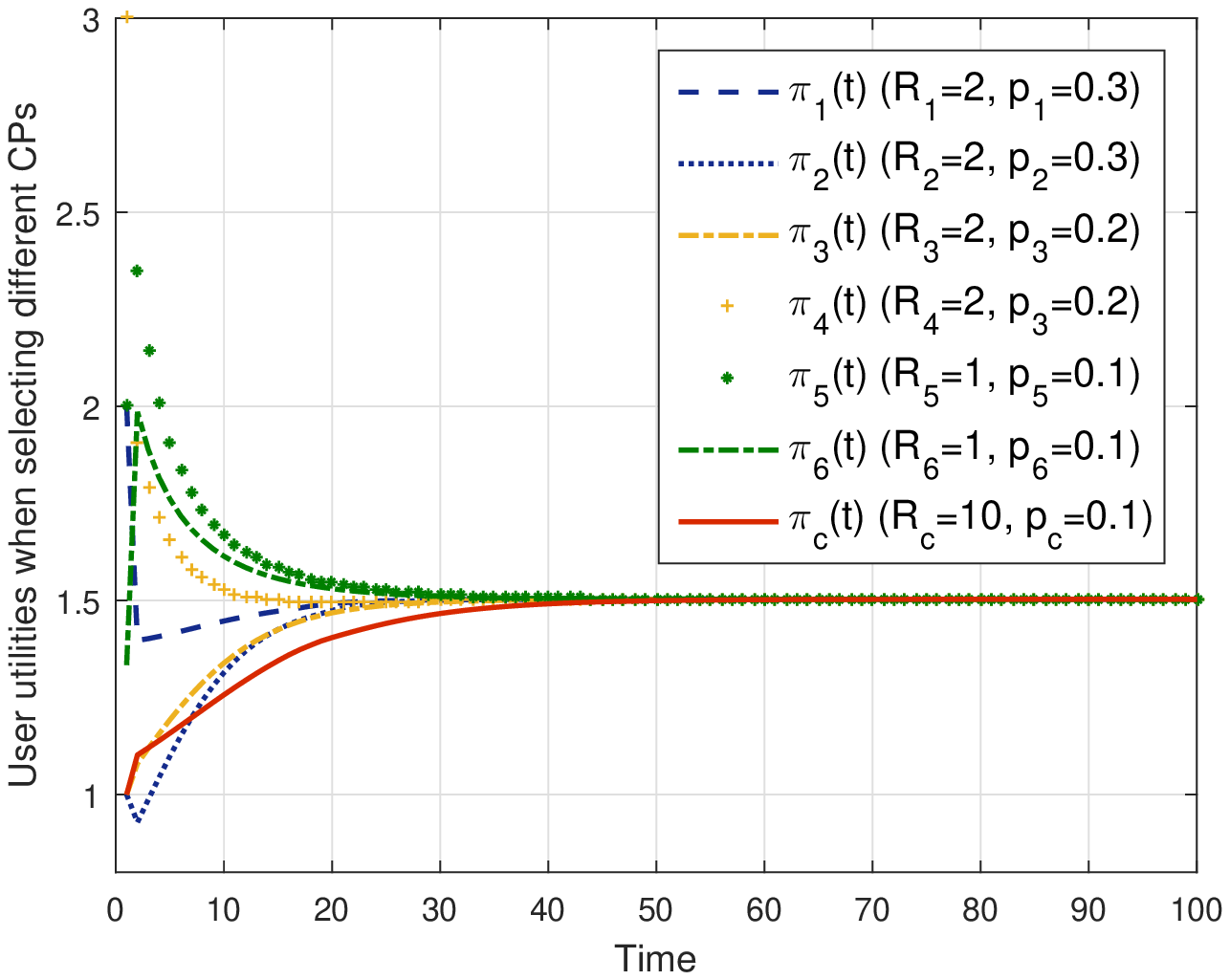}\label{fig:N_pi}}
\caption{Evolutions of population distribution, resource allocation and user utilities versus different $R_n$ and $p_n$ when the number of ECPs is $N=6$.}\vspace{-5mm}
\label{fig:N}
\end{center}
\end{figure}

We also test the population distribution dynamic when there are many ECPs in the ECC system. In particular, consider there are $N=6$ different ECPs selecting the value of their computational power in $\left\{ 1,2 \right\}$ ($kH/s$) and the value of access price in $\left\{ 0.1,0.2,0.3 \right\}$. For the CCP, set ${{R}_{c}}=10kH/s$ and ${{p}_{c}}=0.1$.
In addition, the initial population distribution is set as $\mathbf{x}=\left[ 0.05,0.1,0.15,0.05,0.1,0.15,0.4 \right]$.
Then we get the evolution of population distribution, cloud resource allocation and user utilities, as shown in Fig.~\ref{fig:N}.
As shown in Fig.~\ref{fig:N_x}, one can observe that the proportions of users selecting ECPs with the same ${{R}_{n}}$ and ${{p}_{n}}$ simultaneously converge to the same equilibrium from different initial distributions.
Moreover, results in Fig.~\ref{fig:N_x} also present that the proportion of users selecting ECP 5 and ECP 6 are the highest among all ECPs, which indicates that users are more willing to select the ECPs with lower access price.
Next, we investigate how the local computing capacity and access price affect the cloud resource allocation in the ECC system. As shown in Fig.~\ref{fig:N_r}, ECP 5 and ECP 6 request and receive the most cloud computing shares among the six ECPs, and ECP 1 and ECP 2 are allocated the least. Combining the results in Fig.~\ref{fig:N_x}, results in Fig.~\ref{fig:N_r} indicate that ECPs with more population shares tend to request and receive more computing resource form the CCP, which can increase the utilities obtained by the users selecting these ECPs, as formulated in (\ref{equ:payoff_user}), and meanwhile boost the utilities of both the CCP and ECPs.
Furthermore, results in Fig.~\ref{fig:N_pi} validate that through the replicator dynamics, all user devices will reach the same individual utility at the equilibrium.

\begin{figure}[!t]
  \centering
  \includegraphics[width=0.479\textwidth]{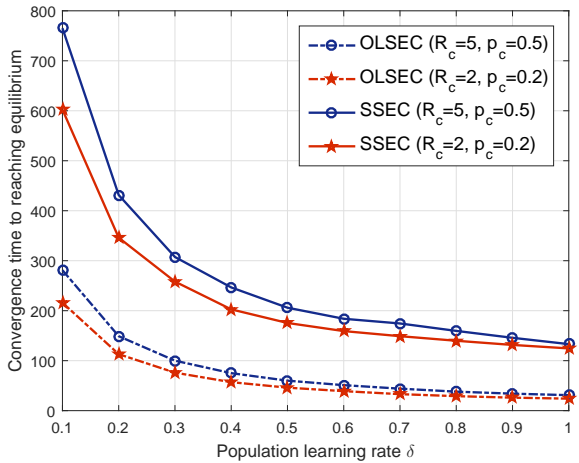}
  \caption{The convergence time versus increasing learning rate $\delta$ and different Stackelberg game equilibrium control schemes.}
  \label{fig:delta}
\end{figure}

Fig.~\ref{fig:delta} illustrates the influence of user learning rate $\delta $ on the convergence speed of evolutionary game and Stackelberg differential game towards the equilibrium. As defined in (\ref{equ:population_dynamic}), learning rate $\delta $ controls the frequency of strategy adaption of all users, which will further control the speed of convergence from initial states towards equilibrium. Results shown in Fig.~\ref{fig:delta} validate that the convergence speed of replicator dynamics grows with the learning rate increasing.
In this part of simulation, we also introduce a classic static Stackelberg equilibrium control (SSEC) proposed in \cite{Chen2012Utility} and \cite{Clempner2018Solving} to optimize the resource pricing and allocation strategies. In SSEC, the CCP and ECPs make their decisions only based on the users' selection strategies, but without the considering of dynamic pricing and allocation strategies among CPs. Then results in Fig.~\ref{fig:delta} indicate that the open-loop Stackelberg equilibrium control (OLSEC) applied in this work can receive a faster convergence speed than SSEC, resulting from the dynamic learning and prediction of all CPs' strategies.

\subsection{Dynamic Pricing and Allocation of Computing Resource }

\begin{figure}[!t]
\begin{center}
\subfigure[Price $p\left(t\right)$]{\includegraphics[width=0.479\textwidth]{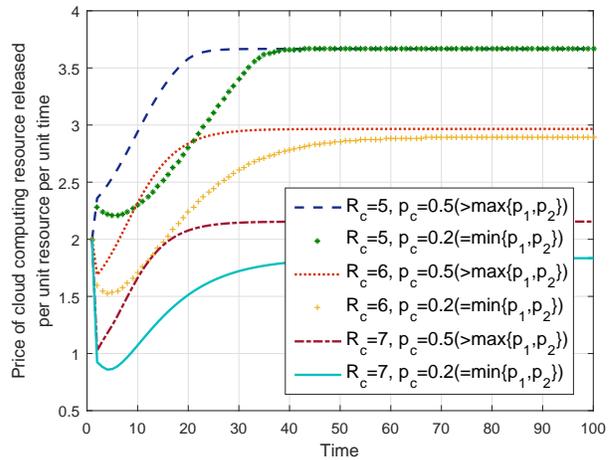} \label{fig:price1}}
\subfigure[Resource reservation $r_c\left(t\right)$]{\includegraphics[width=0.479\textwidth]{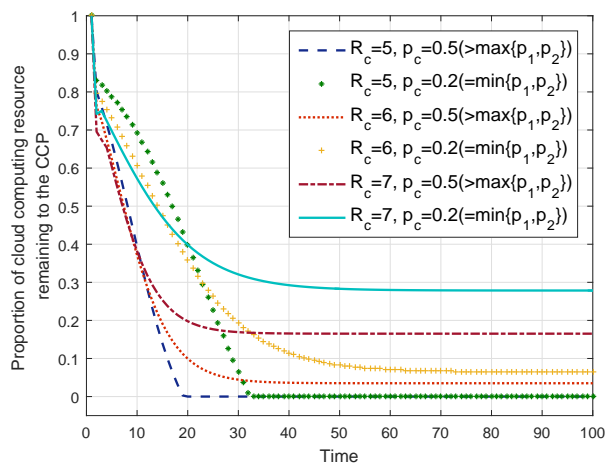}\label{fig:price2}}
\caption{Dynamic computing resource pricing and allocation in the ECC system versus different $R_c$ and $p_c$.}\vspace{-5mm}
\label{fig:price}
\end{center}
\end{figure}

To validate the performance of proposed Stackelberg differential game based resource pricing and allocation strategies, and investigate the impact of ${{R}_{c}}$ and ${{p}_{c}}$ on these strategies made in the computing resource level, we still consider the situation where there are two ECPs in the ECC system. Set ${{R}_{1}}=2kH/s$, ${{R}_{2}}=1kH/s$, ${{p}_{1}}=0.3$ and ${{p}_{2}}=0.2$, which are the same as the simulation in Section~\ref{simu1}.
Let ${{R}_{c}}$ select values in $\left\{ 5,6,7 \right\}$ ($kH/s$) and ${{p}_{c}}$ choose values in $\left\{ 0.5>\max \left\{ {{p}_{1}},{{p}_{2}} \right\},0.2=\min \left\{ {{p}_{1}},{{p}_{2}} \right\} \right\}$.

By applying the evolutionary game based service selection and the Stackelberg differential game based resource pricing and allocation, we obtain the unit price of cloud computing resource and the proportion of could computing resource remaining to the CCP, which are shown in Fig.~\ref{fig:price1} and Fig.~\ref{fig:price2}, respectively.
In Fig.~\ref{fig:price1}, results illustrate that the optimal price at equilibrium decreases with increasing total computational power of CCP.
Meanwhile, the proportion of cloud computing resource remaining to the CCP at equilibrium increases with growing $R_c$.
In addition, results in Fig~\ref{fig:price1} also indicate that with the same $R_c$, the optimal price at equilibrium set by the CCP with lower access price $p_c$ is lower than that with high $p_c$.
Combining the results in Fig.~\ref{fig:price2}, this trend implies that the utility of CCP with lower access price $p_c$ can be optimized by remaining more cloud computing, which will attract more users selecting the CCP, meanwhile setting a high $p_c$ to reduce the ECPs' willingness of purchasing cloud computing resource.
In addition, results in Fig.~\ref{fig:price} can also validate that the strategies of computing resource pricing and allocation will converge to the Stackelberg equilibrium.

\subsection{Influence of Delay in Replicator Dynamics}

\begin{figure}[!t]
  \centering
  \includegraphics[width=0.479\textwidth]{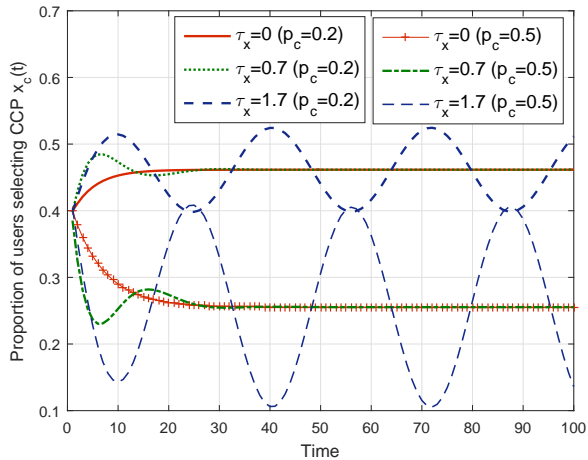}\vspace{-2mm}
  \caption{Proportions of users selecting the CCP under different population delays $\tau_x$ and access price of CCP $p_c$.}\vspace{-3mm}
  \label{fig:delay}
\end{figure}

Next, we study that how the population distribution states change with delay in replicator dynamics.
As defined in (\ref{equ:payoff_user}), the utilities of user devices obtained by selecting different CPs are depends on the service selection strategies of all users in the evolutionary game. However, the information of population distribution state is always delayed resulting from the communication latency. Let ${{\tau }_{x}}\ge 0$ denote the delay of population information. Then the delayed replicator dynamics based on (\ref{equ:population_dynamic}) can be given by
\begin{eqnarray}
\begin{aligned}
{{{\dot{x}}}_{n/c}}\left( t \right)=  & \delta {{x}_{n/c}}\left( t-{{\tau }_{x}} \right)\left[ \pi \left( n/c,\mathbf{x}\left( t-{{\tau }_{x}} \right),\mathbf{r}\left( t \right) \right) \right. \\
 & ~~~~~~~~~~~~~~~~~\left. -\pi \left( \mathbf{x}\left( t \right),\mathbf{x}\left( t-{{\tau }_{x}} \right),\mathbf{r}\left( t \right) \right) \right] ,
\end{aligned}\label{equ:delay:tao}
\end{eqnarray}
where $n\in \mathsf{\mathcal{N}}$.
Considering that delayed replicator dynamics (\ref{equ:delay:tao}) can be rewritten as $\mathbf{\dot{x}}\left( t \right)=\mathbf{Ax}\left( t-{{\tau }_{x}} \right)+\mathbf{b}$, then its characteristic equation can be given by $\Theta {{e}^{-\gamma {{\tau }_{x}}}}+\gamma =0$, where $\Theta $ has been defined in (\ref{equ:Theta}). Here we introduce the necessary and sufficient condition for the stability of delayed replicator dynamics proposed in~\cite{Gopalsamy1992Stability}, which can be given by ${{\tau }_{x}}<{\pi }/{2\Theta }$. Therefore, the stable ESS can be guaranteed with a small population delay.
In this simulation, we set ${{\tau }_{x}}=0.7$ and ${{\tau }_{x}}=1.7$ to test different levels of population delay.
Other parameters are set as $R_1=2kH/s$, $R_2=1kH/s$, $R_c=2$, $p_1=0.3$, $p_2=0.2$ and $p_c=0.2$.
By applying the delayed replicator dynamics, the proportions of users selecting the CCP are shown in Fig.~\ref{fig:delay}. Results in Fig.~\ref{fig:delay} validate that the population distribution state can still converge to the equilibrium after dynamic fluctuation, when ${{\tau }_{x}}$ is small. On the contrary, when ${{\tau }_{x}}$ is large, the equilibrium cannot be reached.

\section{Conclusion}
\label{conculsion}

In this paper, an SDN-based architecture has been established for edge and cloud computing services in 5G wireless HetNets, which can support efficient and on-demand computing resource management to optimize resource utilization and complete the time-varying computational tasks uploaded by user devices.
In addition, considering the incompleteness of information, an evolutionary game based service selection was designed for users, which can model users' replicator dynamics of service subscription when they request the CCP or ECPs for computing resource.
To complete these time-varying computational tasks from users, a Stackelberg differential game based cloud computing resource sharing mechanism was proposed to facilitate the resource trading between the CCP and different ECPs.
Moreover, open-loop Stackelberg equilibrium solutions for the CCP (leader) and ECPs (followers), i.e., the optimal resource pricing and allocation strategies, were derived and obtained, which can promise the maximum integral utilities of the leader and followers over the time horizon, respectively.
Simulation results have validated the performance of the designed resource sharing mechanism, and revealed the convergence and stable states of user selection, resource pricing, and resource allocation in the ECC system.

%
%

\bibliographystyle{IEEEtran}
\bibliography{ref}

\begin{thebibliography}{10}
\providecommand{\url}[1]{#1}
\csname url@samestyle\endcsname
\providecommand{\newblock}{\relax}
\providecommand{\bibinfo}[2]{#2}
\providecommand{\BIBentrySTDinterwordspacing}{\spaceskip=0pt\relax}
\providecommand{\BIBentryALTinterwordstretchfactor}{4}
\providecommand{\BIBentryALTinterwordspacing}{\spaceskip=\fontdimen2\font plus
\BIBentryALTinterwordstretchfactor\fontdimen3\font minus
  \fontdimen4\font\relax}
\providecommand{\BIBforeignlanguage}[2]{{%
\expandafter\ifx\csname l@#1\endcsname\relax
\typeout{** WARNING: IEEEtran.bst: No hyphenation pattern has been}%
\typeout{** loaded for the language `#1'. Using the pattern for}%
\typeout{** the default language instead.}%
\else
\language=\csname l@#1\endcsname
\fi
#2}}
\providecommand{\BIBdecl}{\relax}
\BIBdecl

\bibitem{jiang2016machine}
C.~Jiang, H.~Zhang, Y.~Ren, Z.~Han, K.-C. Chen, and L.~Hanzo, ``Machine
  learning paradigms for next-generation wireless networks,'' \emph{IEEE
  Wireless Commun.}, vol.~24, no.~2, pp. 98--105, Apr. 2017.

\bibitem{nguyen2019market}
D.~T. Nguyen, L.~B. Le, and V.~K. Bhargava, ``A market-based framework for
  multi-resource allocation in fog computing,'' \emph{IEEE/ACM Trans.
  Networking}, vol.~27, no.~3, pp. 1151--1164, Jun. 2019.

\bibitem{yang2019communication}
C.-S. Yang, R.~Pedarsani, and A.~S. Avestimehr, ``Communication-aware
  scheduling of serial tasks for dispersed computing,'' \emph{IEEE/ACM Trans.
  Networking}, vol.~27, no.~4, pp. 1330--1343, Aug. 2019.

\bibitem{chen2015efficient}
X.~Chen, L.~Jiao, W.~Li, and X.~Fu, ``Efficient multi-user computation
  offloading for mobile-edge cloud computing,'' \emph{IEEE/ACM Trans.
  Networking}, vol.~24, no.~5, pp. 2795--2808, Oct. 2016.

\bibitem{jiao2019auction}
Y.~Jiao, P.~Wang, D.~Niyato, and K.~Suankaewmanee, ``Auction mechanisms in
  cloud/fog computing resource allocation for public blockchain networks,''
  \emph{IEEE Trans. Parallel Distrib. Syst.}, vol.~30, no.~9, pp. 1975--1989,
  Sept. 2019.

\bibitem{du2020learningVTM}
J.~Du, C.~Jiang, J.~Wang, Y.~Ren, and M.~Debbah, ``Machine learning for 6{G}
  wireless networks: Carry-forward-enhanced bandwidth, massive access, and
  ultrareliable/low latency,'' \emph{IEEE Veh. Technol. Mag}, vol.~15, no.~4,
  pp. 123--134, Dec. 2020.

\bibitem{premsankar2018edge}
G.~Premsankar, M.~Di~Francesco, and T.~Taleb, ``Edge computing for the internet
  of things: A case study,'' \emph{IEEE Internet Things J.}, vol.~5, no.~2, pp.
  1275--1284, Apr. 2018.

\bibitem{jovsilo2018decentralized}
S.~Jo{\v{s}}ilo and G.~D{\'a}n, ``Decentralized algorithm for randomized task
  allocation in fog computing systems,'' \emph{IEEE/ACM Trans. Networking},
  vol.~27, no.~1, pp. 85--97, Feb. 2019.

\bibitem{liu2018distributed}
Y.~Liu, F.~R. Yu, X.~Li, H.~Ji, and V.~C. Leung, ``Distributed resource
  allocation and computation offloading in fog and cloud networks with
  non-orthogonal multiple access,'' \emph{IEEE Trans. Veh. Technol.}, vol.~67,
  no.~12, pp. 12\,137--12\,151, Dec. 2018.

\bibitem{du2021oceanSHS}
J.~Du, J.~Song, Y.~Ren, and J.~Wang, ``Convergence of broadband and
  broadcast/multicast in maritime information networks,'' \emph{Tsinghua Sci.
  Techno.}, vol.~26, no.~5, pp. 592--607, 2021.

\bibitem{laaroussi2018service}
Z.~Laaroussi, R.~Morabito, and T.~Taleb, ``Service provisioning in vehicular
  networks through edge and cloud: an empirical analysis,'' in \emph{2018 IEEE
  Conf. Standards Commun. Networking (CSCN)}.\hskip 1em plus 0.5em minus
  0.4em\relax Paris France, 29-31 Oct. 2018.

\bibitem{chaudhary2017network}
R.~Chaudhary, N.~Kumar, and S.~Zeadally, ``Network service chaining in fog and
  cloud computing for the 5{G} environment: data management and security
  challenges,'' \emph{IEEE Commun. Mag.}, vol.~55, no.~11, pp. 114--122, Nov.
  2017.

\bibitem{cao2017enhancing}
Z.~Cao, S.~S. Panwar, M.~Kodialam, and T.~Lakshman, ``Enhancing mobile networks
  with software defined networking and cloud computing,'' \emph{IEEE/ACM Trans.
  Networking}, vol.~25, no.~3, pp. 1431--1444, Jun. 2017.

\bibitem{du2018auction}
J.~Du, C.~Jiang, H.~Zhang, Y.~Ren, and M.~Guizani, ``Auction design and
  analysis for {SDN}-based traffic offloading in hybrid satellite-terrestrial
  networks,'' \emph{IEEE J. Sel. Areas Commun.}, vol.~36, no.~10, pp.
  2202--2217, Oct. 2018.

\bibitem{liang2018enhancing}
C.~Liang, Y.~He, F.~R. Yu, and N.~Zhao, ``Enhancing video rate adaptation with
  mobile edge computing and caching in software-defined mobile networks,''
  \emph{IEEE Trans. Wireless Commun.}, no.~10, pp. 7013--7026, Oct. 2018.

\bibitem{kiani2017toward}
A.~Kiani and N.~Ansari, ``Toward hierarchical mobile edge computing: An
  auction-based profit maximization approach,'' \emph{IEEE Internet Things J.},
  vol.~4, no.~6, pp. 2082--2091, Dec. 2017.

\bibitem{zhang2019parking}
Y.~Zhang, C.-Y. Wang, and H.-Y. Wei, ``Parking reservation auction for parked
  vehicle assistance in vehicular fog computing,'' \emph{IEEE Trans. Veh.
  Technol.}, vol.~PP, no.~99, pp. 1--1, Feb. 2019.

\bibitem{wang2018knowledge}
R.~Wang, J.~Yan, D.~Wu, H.~Wang, and Q.~Yang, ``Knowledge-centric edge
  computing based on virtualized {D}2{D} communication systems,'' \emph{IEEE
  Communications Magazine}, vol.~56, no.~5, pp. 32--38, May 2018.

\bibitem{du2017contractJSAC}
J.~Du, E.~Gelenbe, C.~Jiang, H.~Zhang, and Y.~Ren, ``Contract design for
  traffic offloading and resource allocation in heterogeneous ultra-dense
  networks,'' \emph{IEEE J. Sel. Areas Commun.}, vol.~35, no.~11, pp.
  2457--2467, Nov. 2017.

\bibitem{Aujla2018Optimal}
G.~S. Aujla, N.~Kumar, A.~Y. Zomaya, and R.~Rajan, ``Optimal decision making
  for big data processing at edge-cloud environment: An {SDN} perspective,''
  \emph{IEEE Trans. Ind. Inform.}, vol.~14, no.~2, pp. 778--789, Feb. 2018.

\bibitem{Zhang2017Computing}
H.~Zhang, X.~Yong, S.~Bu, D.~Niyato, R.~Yu, and H.~Zhu, ``Computing resource
  allocation in three-tier {I}o{T} fog networks: a joint optimization approach
  combining stackelberg game and matching,'' \emph{IEEE Internet Things J.},
  vol.~4, no.~5, pp. 1204--1215, Oct. 2017.

\bibitem{Xiong2017Edge}
Z.~Xiong, S.~Feng, D.~Niyato, W.~Ping, and H.~Zhu, ``Edge computing resource
  management and pricing for mobile blockchain,'' \emph{IEEE Internet Things
  J.}, vol. Early Access, 24 Sept. 2018.

\bibitem{8004162}
G.~S. {Aujla}, R.~{Chaudhary}, N.~{Kumar}, J.~J. P.~C. {Rodrigues}, and
  A.~{Vinel}, ``Data offloading in 5{G}-enabled software-defined vehicular
  networks: A stackelberg-game-based approach,'' \emph{IEEE Commun. Mag.},
  vol.~55, no.~8, pp. 100--108, Aug. 2017.

\bibitem{8880515}
T.~{Sanguanpuak}, D.~{Niyato}, N.~{Rajatheva}, and M.~{Latva-aho}, ``Radio
  resource sharing and edge caching with latency constraint for local 5{G}
  operator: Geometric programming meets stackelberg game,'' \emph{IEEE Trans.
  Mobile Computing}, pp. 1--1, Early Access, Oct. 2019.

\bibitem{7949095}
G.~S. {Aujla}, M.~{Singh}, N.~{Kumar}, and A.~Y. {Zomaya}, ``Stackelberg game
  for energy-aware resource allocation to sustain data centers using {RES},''
  \emph{IEEE Transactions on Cloud Computing}, vol.~7, no.~4, pp. 1109--1123,
  Oct.-Dec. 2019.

\bibitem{Feng2017Dynamic}
S.~Feng, Z.~Xiong, D.~Niyato, and P.~Wang, ``Dynamic resource management to
  defend against advanced persistent threats in fog computing: A game theoretic
  approach,'' \emph{IEEE Trans. Cloud Computing}, vol. Early Access, 31 Jan.
  2019.

\bibitem{Kim2018An}
S.~H. Kim, S.~Park, C.~Min, and C.~H. Youn, ``An optimal pricing scheme for the
  energy efficient mobile edge computation offloading with {OFDMA},''
  \emph{IEEE Commun. Lett.}, vol.~22, no.~9, pp. 1922--1925, Sept. 2018.

\bibitem{molina2018enhancing}
A.~Molina~Zarca, J.~Bernal~Bernabe, I.~Farris, Y.~Khettab, T.~Taleb, and
  A.~Skarmeta, ``Enhancing {I}o{T} security through network softwarization and
  virtual security appliances,'' \emph{Int. J. Network Manage.}, vol.~28,
  no.~5, p. e2038, Jul. 2018.

\bibitem{farris2018survey}
I.~Farris, T.~Taleb, Y.~Khettab, and J.~Song, ``A survey on emerging {SDN} and
  {NFV} security mechanisms for {I}o{T} systems,'' \emph{IEEE Commun. Surveys
  \& Tutorials}, vol.~21, no.~1, pp. 812--837, Firstquarter 2019.

\bibitem{AUJLA20181279}
G.~S. {Aujla} and N.~{Kumar}, ``{ME}n{S}u{S}: An efficient scheme for energy
  management with sustainability of cloud data centers in edge–cloud
  environment,'' \emph{Future Generation Computer Syst.}, vol.~86, pp.
  1279--1300, Sept. 2018.

\bibitem{Bruschi2019A}
R.~Bruschi, F.~Davoli, P.~Lago, and J.~F. Pajo, ``A multi-clustering approach
  to scale distributed tenant networks for mobile edge computing,'' \emph{IEEE
  J. Sel. Areas Commun.}, vol.~37, no.~3, pp. 499--514, Mar. 2019.

\bibitem{chen2018joint}
Q.~Chen, F.~R. Yu, T.~Huang, R.~Xie, J.~Liu, and Y.~Liu, ``Joint resource
  allocation for software-defined networking, caching, and computing,''
  \emph{IEEE/ACM Trans. Networking}, vol.~26, no.~1, pp. 274--287, Feb. 2018.

\bibitem{Baktir2017How}
A.~C. Baktir, A.~Ozgovde, and C.~Ersoy, ``How can edge computing benefit from
  software-defined networking: A survey, use cases \& future directions,''
  \emph{IEEE Commun. Surveys \& Tutorials}, vol.~19, no.~4, pp. 2359--2391,
  Jun. 2017.

\bibitem{8008830}
G.~S. {Aujla}, N.~{Kumar}, A.~Y. {Zomaya}, and R.~{Ranjan}, ``Optimal decision
  making for big data processing at edge-cloud environment: An sdn
  perspective,'' \emph{IEEE Transactions on Industrial Informatics}, vol.~14,
  no.~2, pp. 778--789, Feb. 2018.

\bibitem{Sharma2017DistBlockNet}
P.~K. Sharma, S.~Singh, Y.~S. Jeong, and J.~H. Park, ``Distblocknet: A
  distributed blockchains-based secure {SDN} architecture for {I}o{T}
  networks,'' \emph{IEEE Commun. Mag.}, vol.~55, no.~9, pp. 78--85, Sept. 2017.

\bibitem{Sharma2018Energy}
P.~K. Sharma, S.~Rathore, Y.~S. Jeong, and J.~H. Park, ``Soft{E}dge{N}et: {SDN}
  based energy-efficient distributed network architecture for edge computing,''
  \emph{IEEE Commun. Mag.}, vol.~56, no.~12, pp. 104--111, Dec. 2018.

\bibitem{zhu2014pricing}
K.~Zhu, E.~Hossain, and D.~Niyato, ``Pricing, spectrum sharing, and service
  selection in two-tier small cell networks: A hierarchical dynamic game
  approach,'' \emph{IEEE Trans. Mobile Computing}, vol.~13, no.~8, pp.
  1843--1856, Aug. 2014.

\bibitem{du2018communityTIFS}
J.~Du, C.~Jiang, K.-C. Chen, Y.~Ren, and H.~V. Poor, ``Community-structured
  evolutionary game for privacy protection in social networks,'' \emph{IEEE
  Trans. Inf. Forens. Security}, vol.~13, no.~3, pp. 574--589, Mar. 2018.

\bibitem{taylor1978evolutionary}
P.~D. Taylor and L.~B. Jonker, ``Evolutionary stable strategies and game
  dynamics,'' \emph{Math. Biosci.}, vol.~40, no. 1-2, pp. 145--156, Jul. 1978.

\bibitem{romero2018dynamic}
J.~Romero, A.~Sanchis-Cano, and L.~Guijarro, ``Dynamic price competition
  between a macrocell operator and a small cell operator: A differential game
  model,'' \emph{Wireless Commun. Mobile Computing}, vol. 2018, May 2018.

\bibitem{Friesz2007Dynamic}
T.~L. Friesz, ``Dynamic optimization and differential games,''
  \emph{Alkalmaz.mat.lapok}, vol. 135, no. 1-2, pp. 203--209, 2007.

\bibitem{Ho1970Differential}
Y.~C. Ho, ``Differential games, dynamic optimization, and generalized control
  theory,'' \emph{J. Optimization Theory \& Applicati.}, vol.~6, no.~3, pp.
  179--209, Sept. 1970.

\bibitem{jiang2013joint}
C.~Jiang, Y.~Chen, Y.~Gao, and K.~R. Liu, ``Joint spectrum sensing and access
  evolutionary game in cognitive radio networks,'' \emph{IEEE Trans. Wireless
  Commun.}, vol.~12, no.~5, pp. 2470--2483, May 2013.

\bibitem{8939471}
F.~{Li}, H.~{Yao}, J.~{Du}, C.~{Jiang}, and Y.~{Qian}, ``Stackelberg game-based
  computation offloading in social and cognitive industrial internet of
  things,'' \emph{IEEE Trans. Ind. Inform.}, vol.~16, no.~8, pp. 5444--5455,
  Aug. 2020.

\bibitem{jiang2013renewal}
C.~Jiang, Y.~Chen, K.~R. Liu, and Y.~Ren, ``Renewal-theoretical dynamic
  spectrum access in cognitive radio network with unknown primary behavior,''
  \emph{IEEE J. Sel. Areas Commun.}, vol.~31, no.~3, pp. 406--416, Mar. 2013.

\bibitem{1973On}
M.~Simaan and J.~B. Cruz, ``On the stackelberg strategy in nonzero-sum games,''
  \emph{J. Optimization Theory Applicat.}, vol.~11, no.~5, pp. 533--555, May
  1973.

\bibitem{2001Existence}
G.~Freiling, G.~Jank, and S.~R. Lee, ``Existence and uniqueness of open-loop
  stackelberg equilibria in linear-quadratic differential games,'' \emph{J.
  Optimization Theory Applicat.}, vol. 110, no.~3, pp. 515--544, Sept. 2001.

\bibitem{anderson1981comparison}
G.~M. Anderson, ``Comparison of optimal control and differential game intercept
  missile guidance laws,'' \emph{J. Guidance, Control, Dynamics}, vol.~4,
  no.~2, pp. 109--115, Mar. -- Apr. 1981.

\bibitem{pachter2019toward}
M.~Pachter, E.~Garcia, and D.~W. Casbeer, ``Toward a solution of the active
  target defense differential game,'' \emph{Dynamic Games Applicat.}, vol.~9,
  no.~1, pp. 165--216, Mar. 2019.

\bibitem{Chen2012Utility}
Y.~Chen, Z.~Jin, and Z.~Qian, ``Utility-aware refunding framework for hybrid
  access femtocell network,'' \emph{IEEE Trans. Wireless Commun.}, vol.~11,
  no.~5, pp. 1688--1697, Jun. 2012.

\bibitem{Clempner2018Solving}
J.~B. Clempner and A.~S. Poznyak, ``Solving transfer pricing involving
  collaborative and non-cooperative equilibria in nash and stackelberg games:
  Centralized–decentralized decision making,'' \emph{Computational Econ.},
  no.~1, pp. 1--29, Jul. 2018.

\bibitem{Gopalsamy1992Stability}
K.~Gopalsamy, \emph{Stability and Oscillations in Delay Differential Equations
  of Population Dynamics}.\hskip 1em plus 0.5em minus 0.4em\relax Berlin,
  Germany. Springer, Mar. 1992.

\end{thebibliography}

\end{document}